%% file: arXiv_v3.tex
\documentclass[11pt]{article}
\input{include_arxiv}

\begin{document}
\title{Security of Key-Alternating Ciphers: Quantum Lower Bounds and Quantum Walk Attacks}
 \author{Chen Bai \thanks{cbai1@vt.edu}}
 \author{Mehdi Esmaili \thanks{mesmaili@vt.edu}}
 \author{Atul Mantri  \thanks{atulmantri@vt.edu}}
 \affil{Department of Computer Science, Virginia Tech, USA 24061}
\date{}    

\maketitle

\begin{abstract}
   We study the quantum security of key-alternating ciphers (KAC), a natural multi-round generalization of the Even–Mansour (EM) construction. KAC abstracts the round structure of practical block ciphers (e.g.~AES) as public permutations interleaved with key XORs. The $1$-round KAC or EM setting already highlights the power of quantum superposition access: EM is secure against classical and Q1 adversaries (quantum access to the public permutation), but insecure in the Q2 model (quantum access to both the permutation and the cipher). The security of multi-round KACs remain largely unexplored; in particular, whether the quantum-classical separation extends beyond a single round had remained open.

\begin{itemize}
\item \textbf{Quantum Lower Bounds.} We prove security of the $t$-round KAC against a non-adaptive adversary in both the Q1 and Q2 models.  In the Q1 model, any distinguiser requires $\Omega(2^{\frac{tn}{2t+1}})$ oracle queries to distinguish the cipher from a random permutation, whereas classically any distinguisher needs $\Omega(2^{\frac{tn}{t+1}})$ queries. As a corollary, we obtain a Q2 lower bound of $\Omega (2^{\frac{(t-1)n}{2t}})$ quantum queries. Thus, for $t \geq 2$, the exponential Q1-Q2 gap collapses in the non-adaptive setting, partially resolving an open problem posed by Kuwakado and Morii (2012). Our proofs develop a controlled-reprogramming framework within a quantum hybrid argument, sidestepping the lack of quantum recording techniques for permutation-based ciphers; we expect this framework to be useful for analyzing other post-quantum symmetric primitives. 

\item \textbf{Quantum Key-Recovery Attack.} We give the first non-trivial quantum key-recovery algorithm for $t$-round KAC in the Q1 model. It makes $O(2^{\alpha n})$ queries with $\alpha = \frac{t(t+1)}{(t+1)^2 + 1}$, improving on the best known classical bound of $O(2^{\alpha' n})$ with $\alpha' = \frac{t}{t+1}$. The algorithm adapts quantum walk techniques to the KAC structure.
\end{itemize}

\end{abstract}


\section{Introduction}

Quantum technologies are rapidly advancing and pose significant threats to the security of digital communications~\cite{shor1999polynomial}. Beyond public-key encryption and digital signatures, the post-quantum security of \textit{symmetric} cryptography is a central concern. Quantum algorithms~\cite{grover1996fast,simon1997power} significantly threaten symmetric ciphers by providing quadratic-and in some settings beyond-quadratic-speedups for generic tasks such as exhaustive key search~\cite{grassl2016applying,jaques2020implementing,bonnetain2022beyond}, potentially compromising widely used ciphers such as the AES~\cite{bonnetain2019quantumaes}. The heuristic ``double the key length" requires validation through rigorous lower bounds in appropriate quantum threat models.

We study \textit{key-alternating ciphers} (KAC), a fundamental abstractions capturing the structure of many block ciphers, including AES~\cite{daemen2002design,daemen2001wide,bogdanov2012key}. In its simplest form, a single-round construction is due to Even and Mansour~\cite{even1997construction}. The $t$-round generalization (Figure~\ref{fig:KAC}) alternates public permutations and $n$-bit key additions:
\[
E_k(x) = k_t \oplus P_t(k_{t-1} \oplus P_{t-1}(\ldots P_1(k_0 \oplus x) \ldots)),
\]
where $P_1,\ldots,P_t$ are public permutations on $\{0,1\}^n$ and $k_0,\ldots,k_t \in \bool^n$ are sampled from a distribution $D$ whose marginals are uniform. For example, AES-128 can be viewed as an instance with $t=10$ in which a round-dependent permutation is iterated and $11n$ bits of round-key material are derived pseudorandomly from a single $n$-bit master key; the master key space remains $\{0,1\}^{128}$.

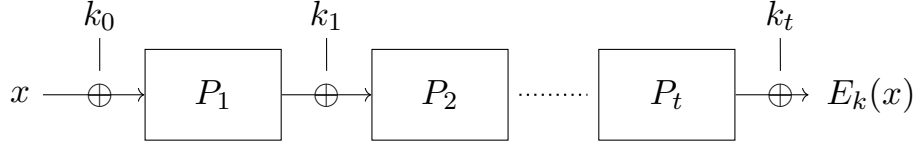
\begin{figure}[ht]
\centering
\begin{tikzpicture}[scale=1.5, every node/.style={scale=1.2}]

\node[draw, minimum width=1.5cm, minimum height=1cm] (P1) at (2,0) {$P_1$};
\node[draw, minimum width=1.5cm, minimum height=1cm] (P2) at (4,0) {$P_2$};
\node[draw, minimum width=1.5cm, minimum height=1cm] (Pt) at (6,0) {$P_t$};

\node (XOR1) at (1,0) {$\oplus$};
\node (XOR2) at (3,0) {$\oplus$};
\node (XOR4) at (7,0) {$\oplus$};

\draw[->] (0.5,0) -- (P1.west);
\draw[->] (P1.east) --  (P2.west); 
\draw[dotted, thick] (4.7,0) -- (5.3,0);

\draw[->] (Pt.east) -- (7.25,0);

\draw[-] (XOR1) -- (1,0.5);
\draw[-] (XOR2) -- (3,0.5);
\draw[-] (XOR4) -- (7,0.5);

\node at (1,0.7) {$k_0$};
\node at (3,0.7) {$k_1$};
\node at (7,0.7) {$k_t$};

\node at (0.3,0) {$x$};
\node at (7.8,0) {$E_k(x)$};

\end{tikzpicture}
\caption{A key-alternating cipher, where Boxes represent public permutations $P_i$'s and $\oplus$ is XOR of the input from left with the keys $k_i$'s.}
\label{fig:KAC}
\end{figure}

To analyze \emph{generic} attacks independent of particular permutation choices, we work in the Random Permutation Model, granting oracle access to the public permutations $P_i$ and to the keyed encryption oracle $E_k$. We distinguish three query regimes: (i) the \emph{Classical} model, where all oracles take classical queries; (ii) the \emph{Q1} model (realistic but mixed access), where the $P_i$ admit superposition queries while $E_k$ is restricted to classical queries; and (iii) the \emph{Q2} model (fully quantum), where both $P_i$ and $E_k$ admit superposition queries. In the classical setting, pseudorandomness of KAC with independent round keys is tightly understood via the H-coefficient technique~\cite{patarin2008coefficients}, yielding $\Theta(2^{\frac{tn}{t+1}})$ distinguishing bounds~\cite{bogdanov2012key,steinberger2012improved,chen2014tight}. 

In the quantum setting, adversaries may issue superposition queries to the public permutations and, depending on the access model, to the keyed encryption oracle as well, making the analysis substantially more delicate than in the classical case. On the provable-security side, classical frameworks such as the H\hbox{-}coefficient technique do not directly transfer: they fundamentally rely on recording full transcripts of adversarial interactions, which cannot be defined or preserved under superposition queries due to no\hbox{-}cloning. A quantum analogue, the compressed-oracle technique~\cite{zhandry2019record}, provides a way to “record’’ queries for \emph{random functions}, but extending it to \emph{permutations} is challenging because of inherent input–output dependencies, despite recent progress~\cite{majenz2025permutation}. Other approaches address permutation-based primitives without relying on the compressed-oracle~\cite{alagic2023two}. For $t=1$ (Even–Mansour), recent work has advanced lower-bound methodology via hybrid arguments that carefully control the trace distance between adjacent hybrids~\cite{alagic2022post,alagic2024post}; however, these analyses are tailored to the Q1 model and the single-round structure, and they do not yet scale to multiple rounds where the adversary gains quantum access to several permutations simultaneously. In particular, there is currently no evidence that KAC remains secure in the strong Q2 model, and determining the round complexity needed to withstand fully quantum adversaries is open.

On the attack side, the landscape for $t=1$ is essentially settled. In Q2, Simon’s algorithm~\cite{simon1997power} yields a polynomial-query break of Even–Mansour~\cite{kuwakado2012security}. In Q1, 
the best-known distinguishers require about $2^{n/3}$ queries, obtained via combinations of Grover search, collision-finding (BHT), and offline Simon's algorithm~\cite{brassard1997quantum,grover1996fast, leander2017grover,bonnetain2019quantum,jaeger2021quantum, kuwakado2012security}. For $t\ge 2$, however, the situation is far less developed: beginning with two rounds, the only published result in Q1~\cite{cai2022quantum} does not outperform classical strategies, and there are no Q2 attacks that leverage coherence across \emph{multiple} public permutations in a way that decisively improves over classical baselines. This gap between the single-round and multi-round cases—both on the side of lower bounds and on the side of concrete algorithms—motivates a more systematic study of quantum query complexity for KAC.

From an algorithmic standpoint, KAC provides a canonical testbed for \emph{mixed-access quantum query complexity}. Many standard quantum speedups (unstructured search, collision, element distinctness, $k$-sum) involve a single oracle, where the distinction between partial and full quantum access is not important. In contrast, cryptographic constructions such as KAC expose multiple oracles by design. This raises a fine-grained question: when a task inherently involves several black-box oracles, how does complexity evolve as only a subset of them admit coherent (quantum) queries? In our setting, “quantum advantage’’ refers to an adversary distinguishing the real KAC world from a random-permutation world using asymptotically fewer queries than any classical strategy. The $t=1$ case already exhibits an exponential advantage between classical/mixed-access (Q1) and fully quantum (Q2) models. Moreover, Q2 is of fundamental interest and it is natural to ask how many rounds are required to withstand fully quantum adversaries and how this threshold interacts with access patterns.

This motivates the following questions for $t$-round KAC in the Random Permutation Model:
\begin{enumerate}
  \item \emph{Provable security against quantum adversaries.} Can we characterize tight $t$-dependent quantum query lower bounds in Q1 and Q2, including the minimal round complexity $t^\star_{\mathrm{Q2}}$ required for security in Q2 model?
  \item \emph{Quantum attacks in realistic models.} In the realistic Q1 model, can we design quantum attacks that strictly outperform classical strategies for $t$-round KAC?
\item \emph{Mixed-access query complexity.} Does granting superposition access to more oracles strictly increase the adversary’s power for $t$-round KAC, or is there a minimum number of quantum access oracles required before any asymptotic advantage over classical appears (if at all)?
\end{enumerate}

\subsection{Our Contributions}
We give the first formal quantum security proof for $t$-round KAC in the \textit{random permutation model with independent permutations and independent round keys}, showing resistance to \textit{non-adaptive} quantum adversaries in both the Q1 and Q2 models. Specifically, in the Q1 model, we prove that distinguishing $t$-KAC from a random permutation requires $\Omega(2^{\frac{tn}{2t+1}})$ oracle queries. Moreover, our Q1 analysis yields a corresponding bound in Q2: any non-adaptive Q2 distinguisher requires $\Omega(2^{\frac{(t-1)n}{2t}})$ quantum queries in the Q2 setting. In particular, for $t \ge 2$, this rules out polynomial-query attacks such as Simon's algorithm, in contrast to the $t=1$ Even-Mansour case~\cite{kuwakado2012security,kaplan2016breaking}. These are obtained via the quantum hybrid method~\cite{Bennett_1997,alagic2022post}, 
together with a new \emph{controlled reprogramming} technique. 

Complementing the lower bounds, we present the first non-trivial \textit{quantum key-recovery} attack applicable to arbitrary $t$-round KAC. This attack operates in the Q1 model (quantum access to the public permutations only). The attack uses $O(2^{\alpha n})$ queries with $\alpha = \frac{t(t+1)}{(t+1)^2 + 1}$, improving over the best-known classical bound $\Theta(2^{\frac{tn}{t+1}})$ of Bogdanov et al.~\cite{bogdanov2012key}. This provides the first generic quantum advantage for key recovery against multi-round KAC in a realistic access model. Algorithmically, the attack is a tailored adaptation of quantum walk techniques given by Magniez et al~\cite{magniez2007search} to the KAC structure.  While quantum walks are among the most powerful tools in quantum algorithms, their use in post-quantum cryptography has been limited~\cite{kaplan2014quantum}; our result demonstrates a concrete advantage in this setting.

Our results provide a comprehensive view of the quantum security landscape for key-alternating ciphers in the non-adaptive setting: explicit lower bounds in Q1 and Q2, and explicit key-recovery attacks in Q1. This partially settles a question of Kuwakado and Morii~\cite{kuwakado2012security}. 

\subsection{Technical Overview}
We now give an informal overview of our main results (see also \expref{Table}{table}), starting with the ideas behind our quantum lower bounds and then outlining our quantum key-recovery attack.

\begin{table}[tb]
\centering
\renewcommand{\arraystretch}{1.2}
\setlength{\tabcolsep}{3pt}
\begin{tabular}{c|c|c|c}
\toprule
Target & Setting & Upper Bound & Lower Bound \\
\midrule
\multirow{3}{*}{\sf KAC} 
  & Classical & $O(2^{n/2})$~\cite{daemen1993limitations} & $\Omega(2^{n/2})$~\cite{even1997construction} \\
  & Q1 & $\tilde{O}(2^{n/3})$~\cite{kuwakado2012security, bonnetain2019quantum} & $\Omega(2^{n/3})$~\cite{alagic2022post} \\
  & Q2 & $O(n)$~\cite{kuwakado2012security} & -- \\
\midrule
\multirow{3}{*}{\sf $2$-KAC} 
  & Classical & $O(2^{2n/3})$ (~\cite{bogdanov2012key}) & $\Omega(2^{2n/3})$ (~\cite{bogdanov2012key}) \\
  & \textcolor{blue}{Q1} & \textcolor{blue}{$O(2^{3n/5})$ ~\magref{Thm}{thm:attack}} & \textcolor{blue}{$\Omega(2^{2n/5})$ ~\magref{Thm}{thm:fullQ1}} \\
  & \textcolor{blue}{Q2} & $O(n2^{n/2})$~\cite{cai2022quantum} & \textcolor{blue}{$\Omega(2^{n/4})$ ~\magref{Cor}{cor:kac-q2}} \\
\midrule
\sf 4-KAC & Q2 (2 keys) & $O(n2^{n/2})$~\cite{anand2024quantum} & -- \\
\midrule
\multirow{3}{*}{\sf $t$-KAC} 
  & Classical & $O(2^{\frac{tn}{t+1}})$~\cite{bogdanov2012key} & $\Omega(2^{\frac{tn}{t+1}})$~\cite{chen2014tight} \\
  & Q2 (same keys/perms) & $O(n)$~\cite{kaplan2016breaking} & -- \\
\midrule
\multirow{1}{*}{\sf $t$-KAC} 
  & \textcolor{blue}{Q1} & \textcolor{blue}{$O(2^{\frac{t(t+1)n}{(t+1)^2+1}})$}~\magref{Thm}{thm:attack} & \textcolor{blue}{$\Omega(2^{\frac{tn}{2t+1}})$}~\magref{Thm}{thm:fullQ1-t} \\
    & \textcolor{blue}{Q2} & -- & \textcolor{blue}{$\Omega(2^{\frac{(t-1)n}{2t}})$}~\magref{Thm}{cor:kac-q2}  \\
\bottomrule
\end{tabular}
\caption{Summary of known upper and lower bounds for KAC under various settings.}
\label{table}
\end{table}

\subsubsection{Lower bound for t-KAC in Q1 model}
We consider the Q1 model, where an adversary has classical access to $E_k$ and quantum access to $P_1,\dots,P_t$ (and their inverse). In this setting, an adversary is able to evaluate the unitary operators $U_{P_i}: \ket{x}\ket{y} \rightarrow \ket{x}\ket{y\oplus P_i(x)}$ on any quantum state it prepares. However, the adversary could only make classical queries to $E_k$, i.e., learning input and output pairs. In this paper, we focus on a non-adaptive adversary that must commit to all of its queries 
(to $E_k$, $P_1,\dots,P_t$) before interacting with the oracles. Once the queries are fixed, the adversary proceeds to query $E_k$, $P_1,\dots, P_t$ and receives the corresponding outputs.

Our work is motivated by the framework of Alagic et al.~\cite{alagic2022post}, 
who applied the hybrid method to establish the post-quantum security of the EM cipher 
in the Q1 model. In contrast, no such results were previously known for multi-round KAC. Here we extend the hybrid approach to prove post-quantum security of $t$-KAC in the Q1 model. 

\medskip
\noindent\textbf{Informal Theorem~1.1}\;\label{thm:infQ1}(Security of $t$-KAC in the Q1 model). \textit{In the Q1 model, any non-adaptive adversary requires $\Omega(2^{\frac{tn}{2t+1}})$ total queries (across classical and quantum queries) to distinguish the t-round Key-Alternating Cipher from a random permutation. More specifically, the distinguishing advantage is bounded by
$\text{Adv}_{\text{$t$-KAC,Q1}}(\A) \leq 4q_{P_1}\cdots q_{P_t}\,\frac{\sqrt{\,q_E}}{2^{tn/2}}$, where $q_E$ denotes the number of classical queries and $q_{P_i}$ the number of quantum queries to $P_i$.}

\medskip
Our proof strategy begins by establishing the post-quantum security of 2-KAC, which serves as the simplest nontrivial case beyond the one-round construction. We then extend the analysis to the general $t$-KAC setting, where the techniques developed for 2-KAC naturally generalize. For the review of techniques, we therefore focus only on outlining the proof idea for the 2-KAC case, as it already captures the essential challenges and methods used in the general setting.
 
\medskip
\noindent\textbf{Security of 2-KAC in the Q1 Model.} 
The $2$-round Key-Alternating Cipher (2-KAC) with key 
$k=(k_0,k_1,k_2)$ from a distribution $D$ is defined on $n$-bit blocks as
\[
E_k(x) \;=\; P_2\!\left(P_1(x \oplus k_0) \oplus k_1\right) \oplus k_2,
\]
with $P_1,P_2 \in \algo P_n$ sampled as independent random permutations on $\{0,1\}^n$. 
We analyze the security of $2$-KAC in the Q1 model, against a non-adaptive adversary.
By Theorem 1.1, the bound we obtain for the Q1 model is
$\text{Adv}_{\text{$2$-KAC,Q1}}(\A) \leq 4q_{P_2}\,q_{P_1}\,\frac{\sqrt{\,q_E}}{2^n}.
$ Here, $q_E$ is the number of classical queries to the block cipher, and $q_{P_1}, q_{P_2}$ are the number of quantum queries to $P_1$ and $P_2$, respectively. Our lower bound shows that any non-adaptive quantum attack requires at least $\Omega(2^{2n/5})$ oracle queries. We now outline the main ideas of our proof at a high level.

We analyze the security of the $2$-KAC construction by considering how many queries 
an adversary must make in order to distinguish whether it is interacting with the 
cipher $E_k$ or with a truly random permutation $R$. 
To formalize this, we introduce a simulator $\mathcal{S}$ that runs one of two experiments with the adversary $\A$. In both, $\mathcal{S}$ samples independent random permutations $P_1,P_2 \gets \Perms_n$; the difference lies in the additional oracle given to $\A$.

\begin{itemize}
    \item \textbf{Real World ($\mathcal{H}_0$):}  
    $\mathcal{S}$ samples a key $k=(k_0,k_1,k_2) \leftarrow D$, 
    where the marginals of $k_0,k_1,k_2$ are uniform over $\{0,1\}^n$, 
    and uses it to construct the cipher $E_k$. 
    It then answers $\A$’s queries using $E_k$, $P_1$, and $P_2$, 
    with classical access to $E_k$ and quantum access to $P_1$ and $P_2$.

    \item \textbf{Ideal World ($\mathcal{H}_2$):}  
    $\mathcal{S}$ samples a random permutation $R \gets \algo P_n$. 
    It then answers $\A$’s queries using $R$, $P_1$, and $P_2$, 
    with classical access to $R$ and quantum access to $P_1$ and $P_2$.
\end{itemize}

The adversary’s goal is to distinguish which world it is operating in. 
We restrict attention to \textit{non-adaptive} adversaries, 
which must fix all of their queries in advance. 
In particular, the adversary must decide on all classical queries to $E$ ($E_k$ in the real world or $R$ in the ideal world)
and all quantum queries to $P_1$ and $P_2$ before any interaction begins. Formally, the adversary prepares the following initial quantum state:
\begin{align*} 
\ket{\Phi_0} \;=\; \sum_{\substack{a_1,\ldots,a_{q_{P_1}}\\ c_1,\ldots,c_{q_{P_2}}}} \alpha_{a_1,\dots,a_{q_{P_1}}} \beta_{c_1,\dots,c_{q_{P_2}}}\; &\ket{x_1,\dots,x_{q_E},a_1,\dots,a_{q_{P_1}}, c_1,\dots,c_{q_{P_2}}}_{XAC}\\
&\otimes\ket{0^{q_E},0^{q_{P_1}},0^{q_{P_2}}}_{YBD}\, \end{align*} 
where the quantum register $X$ corresponds to the classical queries to $E$, 
$A$ and $C$ to the quantum queries to $P_1$ and $P_2$, 
and $Y,B,D$ are the corresponding output registers initialized to zero.  $(x_1,\dots,x_{q_E})$ is a sequence of $q_E$ classical queries, $(a_1,\dots,a_{q_1})$ and $(c_1,\dots,x_{c_2})$ are superpositions of quantum queries to $P_1$ and $P_2$ that $\A$ prepares.
See \expref{Definition}{def:psi0} for details. Our goal is to bound the distinguishing advantage of $\A$ after making 
$q_E$ classical queries to the classical cipher and $q_{P_1}, q_{P_2}$ quantum queries to $P_1$ and $P_2$. 
Formally,
\[
\Adv_{2\text{-KAC},\,Q1}(\A) 
\;=\; 
\left| \Pr[\A(\mathcal{H}_0)=1] - \Pr[\A(\mathcal{H}_2)=1] \right|.
\]

To facilitate the analysis, we introduce an intermediate experiment $\mathcal{H}_1$:

\begin{itemize}
    \item \textbf{Intermediate World ($\mathcal{H}_1$):}  
    $\mathcal{S}$ samples a key $k=(k_0,k_1,k_2)$, 
    where each $k_i$ is marginally uniform over $\{0,1\}^n$, 
    and random permutations $R,P_1,P_2 \gets \Perms_n$. 
    It then answers $\A$’s queries using $R, P_1$ together with a reprogrammed permutation $P_2'$, 
    which agrees with $P_2$ everywhere except on points reprogrammed 
    according to $\A$’s queries to $R,P_1$ and the key $(k_0,k_1)$.
\end{itemize}

By the triangle inequality, the distinguishing advantage can be bounded as
\[
\Adv_{2\text{-KAC},\,Q1}(\A)  
\;\leq\; 
\left| \Pr[\A(\mathcal{H}_0)=1] - \Pr[\A(\mathcal{H}_1)=1] \right| 
+ 
\left| \Pr[\A(\mathcal{H}_1)=1] - \Pr[\A(\mathcal{H}_2)=1] \right|.
\]
The purpose of reprogramming is to ensure that queries are answered consistently 
across $\mathcal{H}_0$ and $\mathcal{H}_1$ so that those hybrids look ``identical" to the adversary's view. To explain how the simulator $\mathcal{S}$ constructs the reprogrammed permutation $P_2'$, we first present the main ideas in the simpler case of 1-KAC: $E_k(x) = P(x \oplus k_0) \oplus k_1$.

\medskip
\noindent \textbf{Warm-up: reprogramming in 1-KAC.} 
Without reprogramming, the two worlds give inconsistent answers: 
in $\mathcal{H}_0$, querying $x$ returns $y = E_k(x) = P(x \oplus k_0) \oplus k_1$, 
whereas in $\mathcal{H}_1$, querying $x$ returns $y = R(x)$, which is independent of $P(x \oplus k_0)$.
An adversary could potentially detect this mismatch and distinguish the two worlds by additionally querying $P$ on input $x\oplus k_0$. To maintain consistency, in $\mathcal{H}_1$ the simulator $\mathcal{S}$ reprograms $P$ as
\[ P' \;=\; P \circ \switch_{\,x \oplus k_0,\, P^{-1}(R(x) \oplus k_1)},
\]
where $\switch_{a,b}$ denotes the transposition swapping $a$ and $b$
so that $P\circ\switch_{a,b}(a)=P(b)$, $P\circ\switch_{a,b}(b)=P(a)$, 
and $P\circ\switch_{a,b}(z)=P(z)$ for all other $z$. This guarantees that $P'(x \oplus k_0) = R(x) \oplus k_1$ in $\Hyb_1$, so the adversary’s query is answered identically in both experiments.

From the 1-KAC case, we see that the reprogramming points are determined solely by 
classical queries. With multiple classical queries, the simulator can construct 
a reprogramming set classically and use it to reprogram $P$. However, in the 
case of 2-KAC this approach no longer works: \textit{reprogramming $P_2$ depends not only 
on classical queries to $R$ but also on quantum queries to $P_1$}. For example, we want to reprogram $P_2$ to $P'_2$ so that 
\[
P'_2(P_1(x\oplus k_0)\oplus k_1)=R(x)\oplus k_2.
\]
There is no known permutation recording technique that can coherently capture superposition queries to $P_1$ 
and use them to update $P_2$ without collapsing the adversary’s state. To address this obstacle, we introduce \emph{controlled reprogramming}. 
The key idea of controlled reprogramming is to update $P_2$ \textit{on-the-fly} 
without ever measuring the adversary’s quantum state. Instead of first 
collecting all reprogramming points classically, we implement the 
reprogramming coherently by applying controlled swap operations to $P_2$ 
whenever a query to $P_1$ reaches a designated reprogramming point. 
This way, consistency is maintained across the hybrids while preserving 
quantum coherence throughout the adversary’s interaction. In a sense, one can view controlled reprogramming as a purification-style technique for the conventional reprogramming. Instead of performing a measurement to determine reprogramming points, we introduce auxiliary registers to compute the swap conditions unitarily. These registers control the swap operations on $P_2$ during the adversary’s query, and are later uncomputed or traced out. This ensures that the overall process  remains coherent and unitary, while effectively achieving the same consistency guarantees as in the classical reprogramming approach. To illustrate this idea, we consider a simpler case where $\A$ only makes one classical query, one quantum query each to $P_1$ and $P_2$.

\begin{tcolorbox}[breakable]
\noindent\textbf{Illustrative Example: 3-Query Non-Adaptive Adversary for 2-KAC.}\par\medskip The initial state that $\A$ prepares is:

\[
|\Phi_0\rangle
\;=\;
\sum_{a,c} \alpha_{a} \beta_{c}\;
|x,a,c\rangle_{XAC}
\;\otimes\;
|0,0,0\rangle_{YBD}\,,
\]
where $\sum_a|\alpha_a|^2=\sum_b|\beta_b|^2=1$. Then, in $\Hyb_0$, and $\Hyb_2$, $\mathcal{S}$ responds with 

\[
|\Psi^{(0)}\rangle 
\;=\;
\sum_{a,c} \alpha_{a} \beta_{c}\;
|x,a,c\rangle_{XAC}
\;\otimes\;
|E_k(x),P_1(a),P_2(c)\rangle_{YBD}\,.
\]

\[
|\Psi^{(2)}\rangle
\;=\;
\sum_{a,c} \alpha_{a} \beta_{c}\;
|x,a,c\rangle_{XAC}
\;\otimes\;
|R(x),P_1(a),P_2(c)\rangle_{YBD}\,.
\]
The non-trivial part is $\Hyb_1$, where $\mathcal{S}$ first answers the query to $R$ and applies $P_1$, producing
\[
\sum_{a,c} \alpha_a \beta_c \;
|x,a,c\rangle_{XAC} \otimes |R(x),P_1(a),0\rangle_{YBD}.
\]
It then coherently computes into fresh ancilla registers $G,U,V$ the ``flag" and ``targets"
\[
g(a)=\mathbf{1}[a=x\oplus k_0], \quad 
X_1 = P_1(x\oplus k_0)\oplus k_1, \quad 
Y_1 = P_2^{-1}(R(x)\oplus k_2),
\]
yielding
\[
\sum_{a,c} \alpha_a \beta_c \;
|x,a,c\rangle_{XAC} \otimes |R(x),P_1(a),0\rangle_{YBD} 
\otimes |g(a)\rangle_{G}|X_1\rangle_{U}|Y_1\rangle_{V}.
\]
The purpose of reprogramming at $(X_1,Y_1)$ is to enforce consistency with the real cipher, 
namely that $P_2'(P_1(x\oplus k_0)\oplus k_1)=R(x)\oplus k_2$. 
To achieve this, $\mathcal{S}$ applies the controlled-swap operator
\[
U_{\mathrm{swap}} = \ket{0}\!\bra{0}_{G}\otimes I 
+ \ket{1}\!\bra{1}_{G}\otimes \Swap_{P_2}(X_1,Y_1),
\]
which swaps entries of $P_2$ at $(X_1,Y_1)$ whenever $G=1$. 
Finally, $\mathcal{S}$ answers the query to the reprogrammed permutation $P'$, uncomputes the ancillas $(G,U,V)$, 
and traces them out, leaving
\[
|\Psi^{(1)}\rangle
= \sum_{a,c} \alpha_a \beta_c \;
|x,a,c\rangle_{XAC} \otimes |R(x),P_1(a),P_2'(c)\rangle_{YBD}.
\]
Here $P_2'$ denotes the permutation obtained from $P_2$ by
\[
P_2' \;=\; P_2 \circ \big(\switch_{X_1,Y_1}\big)^{g(a)}.
\]

that is, swapping the entries at $(X_1,Y_1)$ when $g(a)=1$. Then, our goal is to bound the trace distance between $|\Psi^{(0)}\rangle$ and $|\Psi^{(1)}\rangle$, and between $|\Psi^{(1)}\rangle$ and $|\Psi^{(2)}\rangle$. For the technique overview, we illustrate the idea by bounding the trace distance
\[
\Tr(\Psi^{(1)}, \Psi^{(2)}) 
= \tfrac{1}{2}\big\| \ket{\Psi^{(1)}}\bra{\Psi^{(1)}}-\ket{\Psi^{(2)}}\bra{\Psi^{(2)}} \big\|_1,
\]
in terms of the set of reprogramming points. One can bound $\Tr(\ket{\Psi^{(0)}}, \ket{\Psi^{(1)}})$ analogously. Using triangle inequality, we combine these bounds to get the total distinguishing advantage between the real and the ideal world. Full details can be found in \expref{Section}{sec:lowerbound}.
\end{tcolorbox}

\medskip
\noindent The above example highlights the structural difference we exploit. In $\Hyb_0$ and $\Hyb_2$, the simulator’s joint answer factors across the oracle interfaces: for the three register output $(Y,B,D)$,
\[
|\Psi^{(0)}\rangle \;=\; |\psi_Y^{(0)}\rangle \otimes |\psi_B^{(0)}\rangle \otimes |\psi_D^{(0)}\rangle,
\qquad
|\Psi^{(2)}\rangle \;=\; |\psi_Y^{(2)}\rangle \otimes |\psi_B^{(2)}\rangle \otimes |\psi_D^{(2)}\rangle,
\]
because $E_k(x)$ (or $R(x)$) is independent of $(a,c)$ and the outputs of $P_1$ and $P_2$ depend on disjoint inputs. In $\Hyb_1$, the controlled reprogramming of $P_2$ applies a joint unitary on $(X_1\,Y_1)$ \emph{conditioned} on the $P_1$–branch $a$ with $a=x\oplus k_0$. Consequently, the post-query state
\[
|\Psi^{(1)}\rangle \;=\; \sum_{a,c}\alpha_a\beta_c\,|x,a,c\rangle_{XAC}\otimes|R(x)\rangle_Y\otimes\bigl(|P_1(a)\rangle_B\otimes U_a|P_2(c)\rangle_D\bigr),
\]
with $U_a=\mathrm{Id}$ for $a\ne x\oplus k_0$ and $U_a=\Swap_{P_2}(X_1,Y_1)$ otherwise, does \emph{not} factor with respect to the B-D partition ( $\rho_{BD}^{(1)}\neq \rho_B^{(1)}\otimes \rho_D^{(1)}$), since unitary on $D$ depends coherently on $P_1$ branch.  Thus distinguishing reduces to bounding the trace distance between a product state and an (entangled) state obtained by a controlled unitary supported on the reprogrammed set.

\subsubsection{Lower bound for t-KAC in Q2 model}
We also analyze the security of $t$-round KAC in the strong Q2 model, which allows the adversary quantum access to the cipher in addition to the internal permutations. We provide the first proof of security for $t$-round KAC against non-adaptive adversaries in this model. 

\medskip
\noindent\textbf{Informal Theorem 1.2}~(Security of $t$-KAC in the Q2 model). \textit{Even in the strong Q2 model, where adversaries have quantum access to both the cipher and its internal permutations, a non-adaptive adversary needs at least $\Omega(2^{\frac{(t-1)n}{2t}})$ quantum queries to distinguish $t$-round KAC from a random permutation.}
\medskip

To highlight the main idea, consider a hypothetical adversary $\A_{\text{im}}$ that is given full access to the codebook of $E_k$, i.e., it makes no queries to learn input/output pairs of $E_k$. Hence,
\[
\text{Adv}_{t\text{-KAC},Q2}(\A) \leq \text{Adv}_{t\text{-KAC},Q2}(\A_{\text{im}}) = \text{Adv}_{t\text{-KAC},Q1}(\A_{\text{im}}),
\]
since $\A$ must issue quantum queries to obtain any information about $E_k$. In particular, we have
$
\text{Adv}_{t\text{-KAC},Q1}(\A_{\text{im}}) \leq 4 q_{P_1} \cdots q_{P_t} \cdot \frac{\sqrt{2^n}}{2^{tn/2}}
$
(by setting $q_E = 2^n$). This implies that $\A_{\text{im}}$ must perform $\Omega\!\left(2^{\frac{(t-1)n}{2t}}\right)$ queries to achieve a constant advantage. Consequently, $\A$ also requires at least this many quantum queries.

\subsubsection{Key-Recovery Attack on t-KAC via Quantum Walk.}
Complementing our lower bound results which establish hardness, we also provide a concrete key-recovery attack on the $t$-round KAC in the Q1 model. This attack demonstrates an achievable upper bound on the query complexity required for a successful key recovery.

Our attack leverages the core idea from the classical attack~\cite{bogdanov2012key}, which identifies the correct key by finding multiple input tuples that result in the same key candidate via a set of equations involving $E_k$ and $P_i$. While the classical attack achieves this by exhaustively searching over large sampled sets, incurring a cost of $O(2^{tn/(t+1)})$, our quantum attack significantly speeds up this search process. For simplicity, we first describe the attack model on the 2-KAC cipher, and in Section~\ref{sec:attack} we extend the discussion to the general case of arbitrary $t$.

\begin{figure}[ht]
    \centering
    \includegraphics[scale=.2]{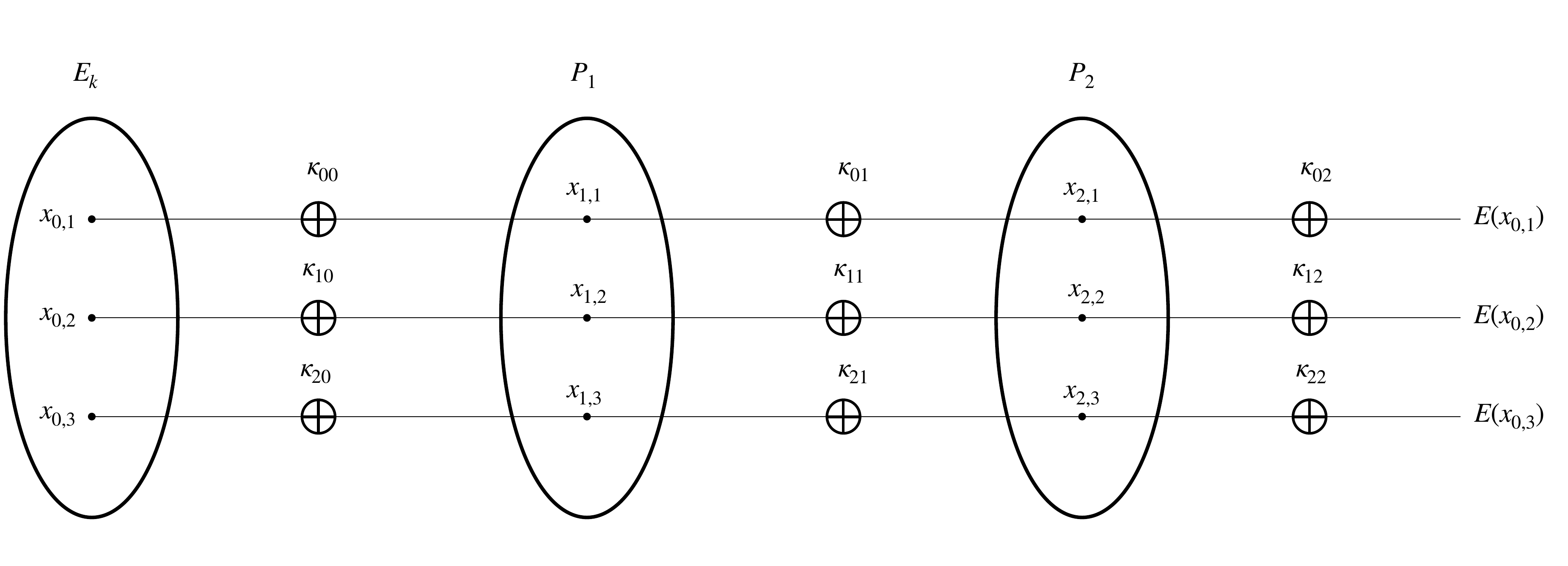}
    \caption{Generating potential key candidates from specific inputs to $E$, $P_1$, and $P_2$. The key candidate $\kappa_0 = (\kappa_{00}, \kappa_{01}, \kappa_{02})$ is generated as $(x_{0,1} \oplus x_{1,1}, P_1(x_{1,1}) \oplus x_{2,1}, P_2(x_{2,1}) \oplus E(x_{0,1}))$. Similarly, $\kappa_1 = (\kappa_{10}, \kappa_{11}, \kappa_{12})$ and $\kappa_2 = (\kappa_{20}, \kappa_{21}, \kappa_{22})$ are generated using other sets of inputs. If $\kappa_0 = \kappa_1 = \kappa_2$, the true key $(k_0, k_1, k_2)$ is highly likely.}
    \label{fig:inputs}
\end{figure}

Since our goal is to recover the cipher’s keys, a natural strategy is to (i) generate all possible key candidates, ensuring that the set consists of correct key, and (ii) identify which candidate corresponds to the actual key. To illustrate how these steps can be carried out, consider three inputs to $E$, $P_1$, and $P_2$, denoted by $x_{0,1}$, $x_{1,1}$, and $x_{2,1}$, respectively, as shown in \expref{Fig}{fig:inputs}. Using these inputs, we can construct a potential key candidate $(\kappa_{00}, \kappa_{01}, \kappa_{02})$ for $(k_0, k_1, k_2)$ as follows:
\[
\kappa_{00}=x_{0,1}\oplus x_{1,1},\quad
\kappa_{01}=P_1(x_{1,1})\oplus x_{2,1},\quad
\kappa_{02}=P_2(x_{2,1})\oplus E(x_{0,1}).
\]

Thus, with a single input to each of $E, P_1, P_2$, we can derive a candidate key for $(k_0, k_1, k_2)$. Extending this to $\alpha$ inputs for each of $E, P_1, P_2$, we obtain $\alpha^3$ potential key candidates. Importantly, if $(\kappa_{00}, \kappa_{01}, \kappa_{02)})$ is a generated key candidate and $(\kappa_{00}, \kappa_{01}) = (k_0, k_1)$, then it necessarily follows that $\kappa_{02} = k_2$. This implies that by generating all possible candidates for $(k_0, k_1)$, one of them must correspond to the correct full key $(k_0, k_1, k_2)$. Constructing all such candidates for $(k_0, k_1)$ requires approximately $2^{2n/3}$ queries to $E, P_1, P_2$.

The next question to address is: How can we determine which key candidate is the correct one? In the classical setting, the solution is relatively straightforward: we can list all the generated candidates and verify which one is correct, requiring only slightly more queries. In the quantum setting, however, the answer is not as simple if we are aiming for a speed-up, since the candidates are prepared in a superposition and any measurement collapses this state to a single candidate, discarding the rest. To tackle this, observe that if we make $\alpha\cdot2^{2n/3}$ queries to $E, P_1, P_2$, the expected number of times each possible key candidate for $(k_0, k_1)$ appears is $\alpha^3$. Now, consider two input tuples $(x_{0,1}, x_{1,1}, x_{2,1})$ and $(x_{0,2}, x_{1,2}, x_{2,2})$ that generate key candidates $(\kappa_{00}, \kappa_{01}, \kappa_{02})$ and $(\kappa_{10}, \kappa_{11}, \kappa_{12})$, respectively, and suppose that $(\kappa_{00}, \kappa_{01}) = (\kappa_{10}, \kappa_{11})$. If $(\kappa_{00}, \kappa_{01})$ is equal to $(k_0, k_1)$, then we necessarily have $\kappa_{02} = \kappa_{12} = k_2$. Otherwise, the equality of $\kappa_{02} = \kappa_{12}$ occurs only with a small probability $1/{2^n}$. Consequently, the true key $(k_0, k_1, k_2)$ will appear approximately $\alpha^3$ times among the generated candidates, whereas such repetitions are exceedingly unlikely for incorrect candidates.

More precisely, if we select inputs to $E, P_1, P_2$ so that each key candidate for $(k_0, k_1)$ is expected to occur about three times, then only a small number of candidates—such as $(\kappa_0, \kappa_1, \kappa_2)$—will actually appear with this multiplicity, and the desired key $(k_0,k_1,k_2)$ is one of them. Defining two input tuples that yield the same key candidate as being \emph{related}, the task of recovering keys reduces to finding three input tuples to $E, P_1, P_2$ that are all mutually related (i.e, they generate the same key candidate). More generally, for $t$-KAC, this corresponds to finding $t+1$ input tuples that share mutual relations.

A wide range of problems can be formulated as finding certain types of relations over sets of function inputs. For instance, in the claw finding problem, the relation is defined for two inputs $x, y$ over functions $f, g$ whenever $f(x) = g(y)$. In the $k$-XOR problem, the relation involves inputs $x_1, \ldots, x_k$ to a function $h$, where we require $h(x_1) \oplus \cdots \oplus h(x_k) = 0$. Over the years, it has been demonstrated that quantum walk algorithms are among the most effective techniques for uncovering such relations~\cite{Tani_2009, ambainis2007quantum, le2014improved, kaplan2014quantum}, and in our case this framework is also directly applicable.

\medskip
\noindent \textbf{Quantum Walk Framework.}
We employ the MNRS-style quantum walk algorithm~\cite{magniez2007search}, which operates on the Cartesian product of $t$ Johnson graphs. In this construction, each vertex corresponds to a collection of $tr$ inputs, selected from the input domains of $P_1, \ldots, P_t$ (with $r$ inputs drawn from each permutation). A vertex $v$ is \emph{marked} if it contains $t+1$ input tuples that can generate the same key candidate. The MNRS framework then allows us to efficiently search for such marked vertices, thereby identifying the right key. To outline how this quantum walk algorithm applies to our setting, we specify the following parameters of the MNRS framework:

\begin{itemize}
    \item \textbf{Setup cost:} $O(tr)$ queries are required to initialize a superposition over the vertices of $G$.
    \item \textbf{Update cost:} Since two vertices in $G$ are adjacent if they share $t(r-1)$ inputs, moving from one vertex to an adjacent vertex requires $O(t)$ queries.
    \item \textbf{Checking cost:} Given all the classical queries together with the $tr$ queries associated with any vertex $v$, the algorithm checks if $v$ is marked without issuing additional queries.

\end{itemize}

By executing the quantum walk algorithm, we can conclude the following theorem:

\medskip
\noindent\textbf{Informal Theorem 3}\;(Quantum Attack). \textit{There exists a quantum algorithm that recovers the keys of a $t$-KAC in the Q1 model with $O(2^{\alpha n})$ classical and quantum queries, where $\alpha = \frac{t(t+1)}{(t+1)^2 + 1}$.}


\subsection{Related Work}

The Even-Mansour (EM) cipher was first introduced in \cite{even1997construction} to construct a block cipher from a publicly known permutation. Later, Daemen extended this idea into the key-alternating cipher (KAC) \cite{daemen2001wide,daemen2002design}, which forms the basis of AES.

Regarding security, in the case of a single round ($t=1$), Even and Mansour showed that an adversary requires about $2^{n/2}$ queries to distinguish $E_k$ from a random permutation. For $t=2$, Bogdanov et al. \cite{bogdanov2012key} proved a security bound of $2^{2n/3}$, while Steinberger \cite{steinberger2012improved} improved the bound to $2^{3n/4}$ for $t=3$. A conjectured general bound of $2^{\frac{tn}{t+1}}$ was eventually proven by Chen and Steinberger \cite{chen2014tight}. These bounds assume independent public permutations and random keys, but \cite{chen2018minimizing,wu2020tight, yu2023security} showed specific constructions that achieve similar security levels.

In the quantum setting, attacks against the $1$-round KAC (Even-Mansour) have been extensively studied, particularly in the Q1 and Q2 models. Kuwakado and Morii \cite{kuwakado2012security} showed that Even-Mansour is completely broken in the Q2 model using Simon’s algorithm. In the Q1 model, they presented an attack with approximately $2^{n/3}$ oracle queries and exponential memory, using the BHT collision-finding algorithm~\cite{brassard1997quantum}. Bonnetain et al. \cite{bonnetain2019quantum2} later improved the Q1 attack to use only polynomial memory by combining Grover’s search with Simon’s algorithm.  

For $t\geq 2$, results are sparse. Recently, Cai et al. \cite{cai2022quantum} provided a Q1 attack on $2$-KAC, which has the same query complexity as the best classical attack, implying no quantum advantage. 

On the quantum lower-bound side, for $t=1$, Jaeger et al.~\cite{jaeger2021quantum} established security for the Even-Mansour cipher against non-adaptive adversaries, and Alagic et al. \cite{alagic2022post,alagic2024post} showed post-quantum security of the Even-Mansour and its tweakable variant against adaptive adversaries. In both adaptive and non-adaptive cases, the tight lower bound is $\Omega(2^{n/3})$, matching the of Kuwakado-Morri~\cite{kuwakado2012security} and Bonnetain et al's~\cite{bonnetain2019quantum} result.  Prior to our work, however, no quantum lower bounds were known for $t$-KAC with $t\ge 2$ in either the Q1 or Q2 models. 

\subsection{Discussion and Future work}
In this work, we make progress in understanding the security of KAC against non-adaptive quantum adversaries in both the Q1 and Q2 model. A natural next step is to establish security against \emph{adaptive} quantum adversaries in both Q1 and Q2, even for $t=2$. Extending hybrid-based arguments and our controlled ``reprogramming" gadget to settings where queries depend on intermediate answers remains the central technical challenge. A second direction is to develop a more fine grained theory of mixed-access models and understand query complexity in those setting: instead of the two endpoints Q1/Q2, parameterize access by the subset $S\subseteq\{P_1,\ldots,P_t,E\}$ of oracles available in superposition and study the resulting complexity $Q^{\mathrm{mix}}(S)$. Does granting quantum access to additional $P_i$ strictly reduce query complexity? In particular, our Q1 attack based on quantum walk continues to apply under \emph{partial} quantum access—e.g., when only a single permutation admits superposition queries—yielding the same asymptotic query bound. This aligns with phenomena in collision finding and $k$-XOR, where optimal quantum speedups do not require that all functions be queried quantumly~\cite{aaronson2004quantum,shi2002quantum,belovs2013adversary}. Formalizing tight lower bounds as a function of the number of quantum-access oracles would clarify where, if anywhere, ``global coherence" across all oracles begins to help.  From a practical perspective, mixed-access models (e.g., $Q^{\mathrm{mix}}(S)$ with a single quantum-access permutation) are closer to NISQ constraints; mapping the precise query advantages and limitations in such models would connect the theory more tightly to realistic quantum attacks.

A complementary set of questions concerns exponents and round dependence. Our non-adaptive lower bounds approach the $2^{n/2}$ barrier as $t\to\infty$; determining whether this $1/2$ exponent is intrinsic for non-adaptive quantum distinguishers, and tightening the gap between upper and lower exponents for small $t$, are concrete open questions. It is also natural to compare tasks: our lower bounds address (in)distinguishability, whereas the upper bound is for key recovery. Establishing black-box reductions or separations between these notions in the quantum setting, particularly under mixed access, would clarify which formulations of security are inherently harder. It's also important to understand how related modeling choices might matter. For instance, allowing inverse oracles ($P_i^{-1}$, $E_k^{-1}$), or limited quantum memory could shift the relevant time–space–query tradeoffs. Finally, several extensions would broaden applicability. Real-world designs use key schedules rather than independent round keys; understanding how mild correlations (e.g., related-key structure) affect the bounds is important. Pushing the techniques beyond KAC to other primitives—Feistel networks, SPNs with small S-boxes, tweakable permutations—and to composition settings as well as different modes of operation and multi-user security remains largely open.

\medskip
\noindent \textbf{Paper Organization.} \expref{Section}{sec:prelim} gives preliminaries. \expref{Section}{sec:lowerbound} and  \expref{Section}{sec:q2lb} shows indistinguishability analysis of $t$-KAC against non-adaptive quantum adversaries in the Q1 and Q2 model, respectively. \expref{Section}{sec:attack} describes our quantum walk attack for $t$-KAC. \expref{Section}{sec:conc} investigates mixed-access query complexity for $t$-KAC, showing that the Q1 attack attains the same asymptotic query complexity even when superposition access is restricted (e.g., to a single permutation). Finally, \expref{Section}{app:special_proofs} provides quantum key-recovery attacks and analyses in special cases, capturing many real-world designs.

\section{Preliminary}
\label{sec:prelim} 
{\bf Notation and basic definitions.} 
We let $\permset{n}$ denote the set of all permutations on~$\bool^n$. In the \emph{public-permutation model} (or random-permutation model), a uniform permutation $P \leftarrow \permset{n}$ is sampled and then provided as an oracle (in both the forward and inverse directions) to all parties.

A block cipher $E: \bool^m \times \bool^n \rightarrow \bool^n$ is a keyed permutation, i.e.,  $E_k(\cdot) = E(k, \cdot)$ is a permutation of~$\bool^n$ for all $k \in \bool^m$. We say $E$ is a (quantum-secure) \emph{pseudorandom permutation} if $E_k$ (for uniform $k \in \bool^m$) is indistinguishable from a uniform permutation in $\permset{n}$ even for (quantum) adversaries who may query their oracle in both the forward and inverse directions.

\medskip
\noindent {\bf Quantum query model.} A quantum oracle is a black-box quantum operation that performs a specific task or computes a specific function. Quantum oracles are often associated with quantum queries, which are used in many quantum algorithms. For a function $f: \bool^n \rightarrow \bool^m$, a quantum oracle $O_f$ is the following unitary transformation: $O_f: \ket{x} \ket{y} \rightarrow \ket{x} \ket{y\oplus f(x)}$, where $\ket{x}$ and $\ket{y}$ represent the states of the input and output registers, respectively. An adversary with quantum oracle access to $f$ means that the adversary can evaluate the above unitary on any quantum states it prepares.

Quantum algorithms typically interact with oracles and their own internal unitaries in an interleaved manner to harness the strengths of each in solving specific computational problems. Consider a $T$-query quantum algorithm starting with an initial state, $\phi_0$ (often an all-zero state), and then alternately applying unitary operators $U_0, \cdots, U_T$, and quantum oracle $O_f$. The final state of the algorithm after $T$ quantum queries is given by $U_TO_fU_{T-1}O_f\cdots U_1O_fU_0\ket{\phi_0}.$ The algorithm produces its output by measuring the final state.

\medskip
\noindent {\bf Quantum distinguishing advantages.} Let $\A$ be a quantum algorithm that makes at most $q$ queries and outputs $0$ or $1$ as the final output, and let $O_1$ and $O_2$ be some oracles. We define the quantum distinguishing advantage of $\A$ by 
\[{\Adv}^{dist}_{O_1,O_2} \coloneqq \left|\Pr_{\substack{O_1}} \! \left[ \A^{O_1}(1^n) = 1 \right] - \Pr_{\substack{O_2}}\! \left[ \A^{O_2}(1^n) = 1 \right]\right|.\] In the case of information theoretic adversaries (i.e. when we only consider the number of queries), we use the notation ${\Adv}^{dist}_{O_1,O_2} \coloneqq max_{\substack{\A}} \{{\Adv}^{dist}_{O_1,O_2}(\A)\} $ where the maximum is taken over all quantum algorithms that make at most $q$ quantum queries. By $R$ we denote the quantum oracle of random permutations, i.e., the oracle such that a permutation $R \in \permset{n}$ is chosen uniformly at random, and adversaries are given oracle access to $R$.

\subsection{Quantum Walk Algorithm}

We start with a regular, undirected graph $G=(V, E)$ with a set of marked vertices $M$. The content of this subsection is adapted from the following papers\cite{ambainis2007quantum,magniez2007search, bonnetain2023finding}.

\begin{definition}\label{def_john}
    The \emph{Johnson graph} $G= (N,r)$ is a connected regular graph that contains $\binom{N}{r}$ vertices such that each vertex is a subset of $\{1,\cdots,N\}$ with $r$ elements. Two vertices $v_1$ and $v_2$ are adjacent if and only if $v_1$ and $v_2$ are only different in one element, i.e., $|v_1 \cap v_2|=r-1$. 
\end{definition}

\begin{definition}\label{def_cat_pr}
Let $G = (V_G, E_G)$ and $H = (V_H, E_H)$ be two graphs. The \emph{Cartesian product} $R = G \times H$ is the graph with vertices $V_R = V_G \times V_H$, where two vertices $(u_1, u_2)$ and $(v_1, v_2)$ are adjacent in $R$ if and only if  $(u_1, v_1) \in E_G$ and $(u_2, v_2) \in E_H$.
\end{definition}

\medskip\noindent{\bf Quantum walk.} 
 Classically, the random walk algorithm looks for a marked vertex in the graph. The algorithm requires associated data $d(v)$ from a data structure $d$ for each $v \in V$ to determine if $v$ contains a marked element. Let $V_d$ be the set of vertices along with the associated data, i.e., $V_d= \{(v,d(v)): v\in V\}$. In the quantum setting, each vertex is a quantum state $\ket{v}$ in a Hilbert space $\algo H$. Then, the quantum analog of $V_d$ is the quantum state $\ket{v}_d=\ket{v}\ket{d(v)}$ in a Hilbert space $\algo H_d$. The quantum walk algorithm will be discussed in the space $\algo H_d$.

In the quantum walk algorithms \cite{ambainis2007quantum,szegedy2004quantum,magniez2007search}, the state space is $C^{|V|\times |V|}$. This means that the walk is on the edges instead of the vertices. At each step, the right end-point of an edge $(v,y)$ is mixed over the neighbors of $v$, and the left end-point is mixed over the neighbors of the new right end-point. Namely, for each vertex $v$, define $\ket{p_{v}}$ as the superposition over all neighbors, i.e. $\ket{p_{v}}= \frac{1}{\sqrt{m}}\sum_{y\in N_{v}}\ket{y}$ and $\ket{p_{v}}_d= \frac{1}{\sqrt{m}}\sum_{y\in N_{v}}\ket{y}_d$, where $m=r(N-r) $ is the number of adjacent vertices. 
Let $A= \{{\sf span}\ket{v}_d\ket{p_v}_d, v\in V\}$ and $B= \{{\sf span}\ket{p_y}_d\ket{y}_d, y\in V\}$. For a state $\ket{\psi}\in \algo H$, let ${\sf ref}_{\psi}=2\ket{\psi}\bra{\psi}-\mathbb{I}$ be the reflection operator. If $\algo S$ is subspace of $\algo H$ spanned by a set of mutually orthogonal states $\{\ket{\psi_i},i\in I\}$, ${\sf ref}(\algo S)= 2\sum_{i\in I}\ket{\psi_i}\bra{\psi_i}-\mathbb{I}$.
Then, the quantum walk operator can be defined as $W_d= {\sf ref}(B)\cdot {\sf ref}(A). $  The first reflection ${\sf ref}(A)$ can be implemented as follows: first move $\ket{v}_d\ket{p_v}_d$ to $\ket{v}_d\ket{0}_d$, perform reflection ${\sf ref}(\algo H_d \otimes \ket{0}_d)$ around $\ket{v}_d\ket{0}_d$ and then reverse the first operation. Followed by the swap operation $\ket{v}_d\ket{y}_d \rightarrow \ket{y}_d\ket{v}_d$,
the second reflection ${\sf ref}(B)$ can be implemented similarly.

\medskip\noindent{\bf MNRS quantum walk framework.} The MNRS framework aims to create a quantum algorithm analogous to the classical walk described above. The goal is to move the initial state $ \ket{\psi}_d =\frac{1}{|V|}\sum_{v\in V}\ket{v}_d\ket{p_v}_d$ to its projection on the target state $ \ket{\phi}_d= \frac{1}{\sqrt{|M|}}\sum_{v\in M}\ket{v}_d\ket{p_{v}}_d.$
Here, $V$ is the set of all vertices and $M$ is the set of marked vertices.  Grover's algorithm says that this can be done by applying the rotation ${\sf ref}(\psi)_d{\sf ref}(\phi^{\bot})_d$ repeatedly, where $\ket{\phi^{\bot}}$ is the state that orthogonal to $\ket{\phi}$\footnote{$\phi^{\bot}$ can be considered as the set of unmarked states}. Since ${\sf ref}(\phi^{\bot})_d = -{\sf ref}(M)_d$, it can be applied with one oracle call. Implementing the operator ${\sf ref}(\psi)_d$ is not trivial and requires phase estimation. In \cite{magniez2007search}, the idea is to design a quantum circuit $R_d$ that simulates ${\sf ref}(\psi)_d$. This quantum circuit needs to call $W_d$ a total of  $O(\frac{1}{\sqrt{\delta}})$ times, where $\delta$ is the spectral gap of $G$. We have omitted the details and would like to direct the interested readers to \cite{magniez2007search}. 

\medskip
\begin{breakablealgorithm1}
\caption{MNRS Quantum Walk Algorithm (See \cite{magniez2007search})}
\label{algo:quantumwalk}
\begin{enumerate}
    \item Initialize the state $\ket{\psi}_d$.
    \item Repeat $O(\frac{1}{\sqrt{\epsilon}})$ times ($\epsilon$ : probability of having a marked vertex):
    \begin{enumerate}
        \item For any basis $\ket{v}_d\ket{y}_d$, flip the phase if $v\in M$.
        \item Apply the operation $R_d$.
    \end{enumerate}
    \item Observe the first register $\ket{v}_d$.
    \item Output $v$ and $d(v)$ if $v\in M$. Otherwise, output ``no marked vertex".
\end{enumerate}

\end{breakablealgorithm1}
\medskip

We analyze the cost of \expref{Algorithm}{algo:quantumwalk}.
\begin{itemize}
 \item {\bf Setup cost} $\sf S$: The cost to prepare the initialization state $\ket{\psi}_d$.
\item {\bf Update cost} $\sf U$: Cost to update the state using operator $R$.
\item {\bf Check cost} $\sf C$: The cost of checking if a vertex $v$ is a marked vertex.
\end{itemize}

\begin{theorem}[Quantum walk (For more details, see \cite{magniez2007search,bonnetain2023finding})]\label{thm:qrw}
   Let $\delta$ be the spectral gap of a graph $G=(V,E)$ and let $\epsilon=\frac{|M|}{|V|}$, where $M$ is the set of marked vertices. Let $\eta=\arcsin(\sqrt{\epsilon})$. There exists a quantum algorithm that can find an element of $M$ with probability at least $1-4\eta^2$.
   The cost is ${\sf S} +\frac{1}{\sqrt{\epsilon}}(\frac{1}{\sqrt{\delta}}{\sf U}+ {\sf C}).$
\end{theorem}

\subsection{Useful Lemmas}

In this section, we provide some lemmas that are used in our proof.

\begin{lemma}(Adapted from \cite{Laszlolecture})\label{lem:sum-capture}
    Let G be a finite abelian group of order $n$, written additively. Let $A_1,\ldots, A_t \subset G$ and let $a$ be a fixed element of $G$. Set $|A_i|=m_i$ for every $i\in \{1,\ldots,t\}$. The expected number of solutions of the equation  $ x_1\oplus \ldots \oplus x_t = a \quad (x_i \in A_i, i=1,\ldots,t )$ is equal to $\frac{m_1\cdot \ldots \cdot m_t}{2^n}$.
\end{lemma}

\begin{lemma}(\cite{Bennett_1997,zhandry2019record})
Let $|\phi\rangle$ and $|\psi\rangle$ be quantum states with Euclidean distance $\epsilon$. Then $\Tr (|\phi\rangle\langle \phi|,|\psi\rangle\langle \psi|) = \epsilon \sqrt{1-\epsilon^{2}/4} \leq \epsilon$.
    
\end{lemma}

\begin{lemma} [Trace distance between pure states, single query]
\label{lem:single}
Let $P$ and $P'$ be two permutations that are identical except on a set $S\subseteq\Dom$.
Let the adversary's initial state be
\[
  \ket{\psi_0}=\sum_{u}\ket{u}\ket{0}\ket{\phi_u},
\]
where $p(u):=\|\phi_u\|^2$ and $\sum_u p(u)=1$.
If an adversary makes a single query to one of these permutations, i.e.,
\[
U_P:\ket{u}\ket{0}\ket{\phi_u}\mapsto \ket{u}\ket{P(u)}\ket{\phi_u},\qquad
U_{P'}:\ket{u}\ket{0}\ket{\phi_u}\mapsto \ket{u}\ket{P'(u)}\ket{\phi_u},
\]
the distinguishability is bounded. The trace distance between the resulting pure states is bounded by:
\[
  \bigl\|\ketbra{\psi_P}{\psi_P}-\ketbra{\psi_{P'}}{\psi_{P'}}\bigr\|_1
  \ \le\ 2\,\sqrt{2\,p(S)}\,,\qquad
  \text{where } p(S):=\sum_{u\in S} p(u).
\]
\end{lemma}
\begin{proof}
After one query,
\(
\ket{\psi_P}=\sum_u \ket{u}\ket{P(u)}\ket{\phi_u}
\)
and
\(
\ket{\psi_{P'}}=\sum_u \ket{u}\ket{P'(u)}\ket{\phi_u}.
\)
Since $P(u)=P'(u)$ for $u\notin S$, orthogonality of distinct basis states implies
\[
\braket{\psi_P|\psi_{P'}}=\sum_{u:\,P(u)=P'(u)} \|\phi_u\|^2
=\sum_{u\notin S} p(u)=1-p(S).
\]
For two pure states,
\(
\bigl\|\ketbra{\psi}{\psi}-\ketbra{\varphi}{\varphi}\bigr\|_1
=2\sqrt{\,1-|\braket{\psi|\varphi}|^2\,}.
\)
Thus,
\[
\bigl\|\ketbra{\psi_P}{\psi_P}-\ketbra{\psi_{P'}}{\psi_{P'}}\bigr\|_1
=2\sqrt{1-(1-p(S))^2}
=2\sqrt{\,p(S)\,(2-p(S))}\ \le\ 2\sqrt{2\,p(S)}.
\]
In the classical case (a point query $u$), $p(S)=\mathbf{1}[u\in S]$, and the
distance is $0$ if $u\notin S$ and $2$ otherwise, as expected.
\end{proof}

The above lemma considers only pure states. Next, we show that the same upper bound also holds when the adversary prepares a mixed state, possibly entangled with an auxiliary register.

\begin{lemma}[Trace distance between mixed states, single query.] \label{lem:tr-ms}
Let $P$ and $P'$ be two permutations that are identical except on a set $S\subseteq\Dom$.
If an adversary makes a single query to one of these permutations,
then for a quantum query with input distribution $p(\cdot)$, the trace distance
between the resulting post-query states satisfies
\[
\bigl\|\sigma_P-\sigma_{P'}\bigr\|_1
\;\le\; 2\sqrt{2\,p(S)}\,,
\qquad \text{where } p(S):=\sum_{u\in S}p(u).
\]
\end{lemma}

\begin{proof}
Let $\mathcal U_P(\cdot)=U_P(\cdot)U_P^\dagger$ and $\mathcal U_{P'}(\cdot)=U_{P'}(\cdot)U_{P'}^\dagger$ be the query channels acting on registers $(U,Y)$, with a clean response $Y$:
\[
U_P:\ket{u}\ket{0}\mapsto \ket{u}\ket{P(u)},\qquad
U_{P'}:\ket{u}\ket{0}\mapsto \ket{u}\ket{P'(u)},
\]
and both act trivially on any reference system $R$.
Let the pre-query state on $(U,Y,R)$ be $\rho_{UYR}$,
and write $\rho_U=\Tr_{YR}\rho_{UYR}$ and
\[
p_\rho(S):=\Tr[\Pi_S\,\rho_U],\qquad \Pi_S:=\sum_{u\in S}\ket{u}\!\bra{u}.
\]
Purify $\rho_{UYR}$ to some $\ket{\Psi}_{UYRR'}$. Expanding in the $U$ basis (using that $Y$ is clean),
\[
\ket{\Psi}=\sum_{u}\ket{u}_U\ket{\eta_u}_{YRR'}
=\sum_{u}\alpha_u\,\ket{u}_U\ket{0}_Y\ket{r_u}_{RR'},
\]
where $p(u)=\|\eta_u\|^2=|\alpha_u|^2$ and $\ket{r_u}:=\eta_u/\alpha_u$ when $p(u)>0$ (arbitrary otherwise).
Define the post-query purifications
\[
\ket{\phi_P}=(I_{RR'}\!\otimes U_P)\ket{\Psi},\qquad
\ket{\phi_{P'}}=(I_{RR'}\!\otimes U_{P'})\ket{\Psi}.
\] 
Since $P(u)=P'(u)$ for $u\notin S$ and the $U$/$Y$ bases are orthonormal,
\[
\braket{\phi_P}{\phi_{P'}}=\sum_{u\notin S}p(u)=1-p_\rho(S).
\]
For pure states,
\(
\bigl\||\alpha\rangle\!\langle\alpha|-|\beta\rangle\!\langle\beta|\bigr\|_1
=2\sqrt{1-|\langle\alpha|\beta\rangle|^2}.
\)
Hence
\[
\bigl\||\phi_P\rangle\!\langle\phi_P|-|\phi_{P'}\rangle\!\langle\phi_{P'}|\bigr\|_1
=2\sqrt{\,p_\rho(S)\,\bigl(2-p_\rho(S)\bigr)}.
\]
Let $\sigma_P=\Tr_{R'}|\phi_P\rangle\!\langle\phi_P|$ and
$\sigma_{P'}=\Tr_{R'}|\phi_{P'}\rangle\!\langle\phi_{P'}|$ be the actual post-query states on $(U,Y,R)$.
By contractivity of trace distance under CPTP maps (here, $\Tr_{R'}$),
\[
\|\sigma_P-\sigma_{P'}\|_1
\;\le\;
\bigl\||\phi_P\rangle\!\langle\phi_P|-|\phi_{P'}\rangle\!\langle\phi_{P'}|\bigr\|_1
=2\sqrt{\,p_\rho(S)\,\bigl(2-p_\rho(S)\bigr)}
\;\le\; 2\sqrt{2\,p_\rho(S)}.
\]
Finally, with a clean response register we have $p_\rho(S)=\sum_{u\in S}p(u)=p(S)$, which yields the stated bound.
\end{proof}

\begin{lemma}[Triangle Inequality for Multiple Queries]
\label{lem:triangle}
If an adversary makes $q$ independent queries to an oracle, and we know the trace distance bound for a single query (from \expref{Lemma}{lem:single} or \expref{Lemma}{lem:tr-ms}), we can bound the total trace distance by simply summing the bounds for each query. This is a direct application of the triangle inequality for the trace distance metric. If changing $P\to P'$ contributes at most $2\sqrt{2p_j(S)}$ for the $j$-th query, then the total distance is bounded by:
\[
  \bigl\|\psi(P)-\psi(P')\bigr\|_1 \ \le\ 2\sum_{j=1}^q \sqrt{2\,p_j(S)\,}.
\]
\end{lemma}

\begin{lemma}
\label{lem:avg}
If the set of reprogrammed points $S\subseteq\Dom$ is a random set chosen independently of the adversary's $j$-th query to $P$, then the expected probability that the query lands on $S$ is given by the expected size of $S$ relative to the total domain size $N$. By linearity of expectation:
\[
  \mathbb{E}\bigl[p^{(P)}_j(S)\bigr] = \sum_{u \in \Dom} p^{(P)}_j(u) \mathbb{E}[\mathbf{1}[u \in S]] = \sum_{u \in \Dom} p^{(P)}_j(u) \frac{\mathbb{E}[|S|]}{N} = \frac{\mathbb{E}[|S|]}{N}.
\]
\end{lemma}

\input{lowerbound}

\input{upperbound}
\input{q1-star}
\input{special_cases}

\section*{Acknowledgement} The authors would like to thank Gorjan Alagic for insightful discussions and reviewers of Eurocrypt and Asiacrypt for constructive feedback on the previous versions of this work. This work is supported by AM’s faculty startup grant from Virginia Tech and, in part, by the Commonwealth Cyber Initiative, an investment in the advancement of cyber research, innovation, and workforce development.

\addcontentsline{toc}{section}{References}
\bibliographystyle{ieeetr}
\bibliography{refs}

\end{document}

%% file: include_arxiv.tex
\usepackage{amsmath,amssymb,amsthm}
\usepackage{algorithm, algorithmicx, algpseudocode}
\usepackage{graphicx}

\usepackage[most]{tcolorbox}
\tcbset{colback=gray!3, colframe=black!40, boxrule=0.4pt, arc=2mm}
\usepackage{enumitem}         
\setlist[itemize]{leftmargin=1.5em, labelsep=0.5em} 
\usepackage{color}
\usepackage{cite}
\usepackage{mathtools}
\usepackage{tikz}
\usepackage{booktabs}
\usepackage{tabularx}
\usepackage{authblk}
\usepackage{geometry}
\usepackage{xcolor}
\usepackage{multirow}
\usepackage{appendix}
\usepackage{comment}
\usepackage{framed}
\usepackage[most]{tcolorbox}
\usepackage{amsthm}
\usepackage{braket}
\usepackage{physics}
\usepackage{subcaption}
\usepackage{quantikz}

\def\bool{\{0,1\}}
\def\A{{\algo A}}

\newcommand{\permset}[1]{{\cal P}(n)}
\newcommand\algo{\mathcal}

\newcommand{\Perms}{\mathcal{P}}

\newcommand{\switch}{{\sf swap}}

\newcommand{\Hyb}{\operatorname{\mathcal{H}}}

 \newcommand{\PrA}[1]{\Pr[\A(\Hyb_{#1})=1]} 
\newcommand{\costS}{\mathsf{S}}
\newcommand{\costU}{\mathsf{U}}

\newcommand{\Dom}{\{0,1\}^n}

\newcommand{\Adv}{\mathsf{Adv}}

\newcommand{\Swap}{\mathrm{SWAP}}
\newcommand{\E}{\mathbb{E}}

\definecolor{YB}{RGB}{40,70,200}  
\definecolor{FLAG}{RGB}{160,0,160} 
\definecolor{Ptwo}{RGB}{0,120,120} 
\definecolor{Dout}{RGB}{200,110,0} 
\definecolor{ANC}{gray}{0.35}

\newcommand{\negl}{\mathsf{negl}}

\usepackage{ifthen}
\newif\ifdraft

\draftfalse 

\ifdraft
\newcommand{\chen}[1]{{\color{orange} #1}}
\newcommand{\atm}[1]{{\color{red} #1}}
\newcommand{\mehdi}[1]{{\color{blue} #1}}
\else
\newcommand{\chen}[1]{}
\newcommand{\atm}[1]{}
\newcommand{\mehdi}[1]{}
\fi

\newtheorem{theorem}{Theorem}[section]

\newtheorem{lemma}[theorem]{Lemma}
\newtheorem{remark}[theorem]{Remark}

\newtheorem{definition}{Definition}[section]
\newtheorem{proposition}{Proposition}

\makeatletter
\newenvironment{breakablealgorithm}
  {  
    \refstepcounter{algorithm}
     \hrule height.8pt depth0pt \kern2pt
     \renewcommand{\caption}[2][\relax]{
       {\raggedright\textbf{Algorithm~\thealgorithm} ##2\par}
       \ifx\relax##1\relax
         \addcontentsline{loa}{algorithm}{\protect\numberline{\thealgorithm}##2}
       \else
         \addcontentsline{loa}{algorithm}{\protect\numberline{\thealgorithm}##1}
       \fi
       \kern2pt\hrule\kern2pt
       \small
     }
  }{
     \kern2pt\hrule\relax
  }
\makeatother

\makeatletter
\newenvironment{breakablealgorithm1}
  {  
    \refstepcounter{algorithm}
     \hrule height.8pt depth0pt \kern2pt
     \renewcommand{\caption}[2][\relax]{
       {\raggedright\textbf{Algorithm~\thealgorithm} ##2\par}
       \ifx\relax##1\relax
         \addcontentsline{loa}{algorithm}{\protect\numberline{\thealgorithm}##2}
       \else
         \addcontentsline{loa}{algorithm}{\protect\numberline{\thealgorithm}##1}
       \fi
       \kern2pt\hrule\kern2pt
       \small
     }
  }{
     \kern2pt\hrule\relax
  }
\makeatother
\definecolor{webgreen}{rgb}{0,.5,0}
\usepackage[colorlinks=true,
			urlcolor=webblue,
			linkcolor=webgreen,
			filecolor=webblue,
			citecolor=webgreen,
			pdfpagemode=UseOutlines,
			pdfstartview=FitH,
			pdfpagelayout=OneColumn,
			bookmarks=true,
			bookmarksdepth=3,
			bookmarksopen=true,
			backref=page]{hyperref}
\newcommand{\expref}[2]{\texorpdfstring{\hyperref[#2]{#1~\ref{#2}}}{#1~\ref{#2}}}
\newcommand{\magref}[2]{\hypersetup{linkcolor=magenta}\hyperref[#2]{#1~\ref{#2}}\hypersetup{linkcolor=webgreen}}

\tcbset{
  simplebox/.style={
    colback=white,
    colframe=black,
    boxrule=0.4pt,
    arc=1pt,
    left=6pt,
    right=6pt,
    top=4pt,
    bottom=4pt,
    fonttitle=\bfseries,
    title={Search Game $\algo{G}^C$ (Classical)}
  }
}

\tcbset{
  simplebox1/.style={
    colback=white,
    colframe=black,
    boxrule=0.4pt,
    arc=1pt,
    left=6pt,
    right=6pt,
    top=4pt,
    bottom=4pt,
    fonttitle=\bfseries,
    title={Search Game $\algo{G}^Q$ (Quantum)}
  }
}

\makeatletter
\newcommand\appendix@section[1]{%
  \refstepcounter{section}%
  \orig@section*{Appendix \@Alph\c@section: #1}%
  \addcontentsline{toc}{section}{Appendix \@Alph\c@section: #1}%
}
\let\orig@section\section
\g@addto@macro\appendix{\let\section\appendix@section}
\makeatother

%% file: lowerbound.tex
\section{Post-Quantum Security of KAC in Q1 model}\label{sec:lowerbound}
We first give a detailed explanation for 2-KAC and later extend the result to $t$ round.

\subsection{Quantum lower bound for 2-KAC}
In this section, we analyze the security of $2$-KAC in the Q1 model against a \textit{non-adaptive} adversary~$\A$. Suppose $\A$ makes $q_E$ classical queries, $q_{P_1}$ quantum queries to $P_1$ and $q_{P_2}$ quantum queries to $P_2$. The formal definition is given below.

\begin{definition}[Non-adaptive adversary in the Q1 model]\label{def:psi0} 
A non-adaptive adversary $\A$ against $2$-KAC in the Q1 model chooses all of its queries before any interaction with the oracles. Concretely, it fixes: 
\begin{itemize} 
\item a sequence of classical queries $(x_1,\dots,x_{q_E})$ to the classical oracle, \item a superposition of quantum queries $(a_1,\dots,a_{q_{P_1}})$ to $P_1$ with amplitudes $\alpha_{\vec{a}}$, 
\item a superposition of quantum queries $(c_1,\dots,c_{q_{P_2}})$ to $P_2$ with amplitudes $\beta_{\vec{c}}$. 
\end{itemize} The initial state prepared by $\A$ is 

\begin{align*} 
\ket{\Phi_0} \;=\; \sum_{\substack{a_1,\ldots,a_{q_{P_1}}\\ c_1,\ldots,c_{q_{P_2}}}} \alpha_{a_1,\dots,a_{q_{P_1}}} \beta_{c_1,\dots,c_{q_{P_2}}}\; &\ket{x_1,\dots,x_{q_E},a_1,\dots,a_{q_{P_1}}, c_1,\dots,c_{q_{P_2}}}_{XAC} 
\\ &\;\otimes\; \ket{0^{q_E},0^{q_{P_1}},0^{q_{P_2}}}_{YBD}\,, \end{align*} 

where 

\begin{itemize} 
\item $X=(X_1,\dots,X_{q_E})$ and $Y=(Y_1,\dots,Y_{q_E})$ are classical registers holding the inputs and outputs of the $q_E$ classical queries, 
\item $A=(A_1,\dots,A_{q_{P_1}})$ and $B=(B_1,\dots,B_{q_{P_1}})$ are quantum registers for the inputs and outputs of the $q_{P_1}$ quantum queries to $P_1$, 
\item $C=(C_1,\dots,C_{q_{P_2}})$ and $D=(D_1,\dots,D_{q_{P_2}})$ are quantum registers for the inputs and outputs of the $q_{P_2}$ quantum queries to $P_2$. 
\end{itemize} 
After preparing $\ket{\Phi_0}$, the adversary proceeds by making its $q_E$ classical queries to the classical oracle, followed by its $q_{P_1}$ quantum queries to $P_1$, and finally its $q_{P_2}$ quantum queries to $P_2$.
\end{definition}

For simplicity, in this paper we use the pure adversarial state without reference register as well as compact notation 
\[
\ket{\Phi_0}
\;=\;
\sum_{\vec{a},\vec{c}} \alpha_{\vec{a}} \beta_{\vec{c}}\;
\ket{\vec{x},\vec{a},\vec{c}}_{XAC}
\;\otimes\;
\ket{0^{q_E},0^{q_{P_1}},0^{q_{P_2}}}_{YBD},
\]
where $\vec{x}=(x_1,\dots,x_{q_E})$, $\vec{a}=(a_1,\dots,a_{q_{P_1}})$, and $\vec{c}=(c_1,\dots,c_{q_{P_2}})$. Later in \expref{Remark}{lbl:rem_mixed}, we explain that our distinguishing results remain valid when $\mathcal{A}$ begins in a mixed state that may be entangled with an auxiliary register.

By non-adaptivity, the choice of $\vec{x}$ and the amplitudes $(\alpha_{\vec{a}},\beta_{\vec{c}})$ is fixed before any oracle interaction, and cannot depend on oracle outputs.

Recall the definition of $2$-KAC :   
$ E_k(x) = k_2 \oplus P_2(k_1 \oplus P_1(k_0 \oplus x)),$ where $k=(k_0,k_1,k_2)$ is sampled from a distribution $D$ such that the marginal distribution of each $k_i\in \bool^n$ is uniform. $P_1, P_2: \bool^n\rightarrow \bool^{n}$ are independent random public permutations. 

The goal of the adversary $\algo A$ is to distinguish between two worlds: the \emph{real} world where $\algo A$ gets access to $(E_k,P_1, P_2)$; the \emph{ideal} world where $\algo A$ gets access to $(R,P_1,P_2)$, where $R$ is a truly random permutation. Our main theorem is stated as follows.

\begin{theorem}[Security of $2$-KAC, Q1 model]\label{thm:fullQ1}
     Let $D$ be a distribution over $k=(k_0,k_1,k_2)$  such that the marginal distributions of $k_0$, $k_1$, and $k_2$ are each uniform, and let $\A$ be a non-adaptive adversary that makes $q_E$ classical queries to the cipher, $q_{P_1}$ quantum queries to $P_1$, and $q_{P_2}$ quantum queries to $P_2$. Let $\mathbf{Adv}_{\text{$2$-KAC,Q1}}(\A)$ denote the qPRP advantage of $\A$:
   \begin{equation*}
    \mathbf{Adv}_{\text{$2$-KAC,Q1}}(\A) := \mathbb{E}_{\substack{k \leftarrow D \\ R, P_1,P_2 \leftarrow \Perms_n }}\left|\Pr \left[\A^{E_k, P_1,P_2}(1^n) = 1\right]
    - \Pr \left[\A^{R, P_1,P_2}(1^n) = 1\right]\right|.
    \end{equation*}
    We obtain $    \mathbf{Adv}_{\text{$2$-KAC,Q1}}(\A) \leq 4q_{P_2}\,q_{P_1}\,\frac{\sqrt{\,q_E}}{2^n}$.
     \end{theorem}

\begin{proof}
Without loss of generality, we assume $\A$ never makes repeated queries. That is, all classical queries $x_1,\dots,x_{q_E}$ are distinct values.
We define a sequence of hybrids $\Hyb_1, \Hyb_2, \Hyb_3$.

\begin{enumerate}
   \item[$\Hyb_0$:] \textbf{Real World.}  
The simulator $\mathcal{S}$ samples $P_1,P_2 \leftarrow \Perms_n$ and $k \leftarrow D$. 
It then prepares the joint state $\ket{\Psi_0}$ after receiving $\ket{\Phi_0}$ from $\A$:
\[
\ket{\Psi_0}\;=\;
\sum_{\vec{a},\vec{c}} \alpha_{\vec{a}} \beta_{\vec{c}}\;
\ket{\vec{x},\vec{a},\vec{c}}_{XAC}
\;\otimes\;
\ket{0^{q_E},0^{q_{P_1}},0^{q_{P_2}}}_{YBD}
\;\otimes\;
\ket{k}_{K}\,\ket{P_1^{\pm}}_{P_1^{\mathrm{reg}}}\,\ket{P_2^{\pm}}_{P_2^{\mathrm{reg}}},
\]
where $X,A,C,Y,B,D$ are as in \expref{Definition}{def:psi0}, 
and $K, P_1^{\mathrm{reg}}, P_2^{\mathrm{reg}}$ are internal registers 
(not accessible to $\A$) holding the sampled key and permutations. 

Here the notation ``$\pm$'' indicates that the adversary has access to both 
forward and inverse evaluations of $P_1$ and $P_2$. Each permutation register 
encodes the full truth table of the permutation in the computational basis, e.g.,
\[
\ket{P_1^{\pm}}_{P_1^{\mathrm{reg}}}
=\bigotimes_{x\in\{0,1\}^n} \ket{x,P_1(x)},
\]
and analogously for $P_2$. This representation enables the simulator to implement 
both forward queries $x\mapsto P_1(x)$ and inverse queries $y\mapsto P_1^{-1}(y)$. 
For simplicity, in the proof we consider forward queries only, 
since the inverse case is handled symmetrically.

$\mathcal{S}$ then answers these queries by applying $E_k, P_1, P_2$ to the corresponding registers, 
and returns the resulting state to $\A$. At the end of the interaction, $\A$’s view is that of an execution with $(E_k, P_1, P_2)$. 

\item[$\Hyb_1$:]\textbf{Intermediate World.}
The simulator $\mathcal{S}$ samples $R, P_1, P_2 \leftarrow \Perms_n$ and a key $k \leftarrow D$, 
and prepares the joint state $\ket{\Psi_0}$ after receiving the query state $\ket{\Phi_0}$ from $\A$. 
In this hybrid experiment, the encryption oracle $E_k$ is replaced by a truly random permutation $R$, 
and the second permutation is replaced by a modified oracle $P_2^{R,P_1,k}$. 
The simulator answers $\A$’s queries by coherently applying $R$, $P_1$, and $P_2^{R,P_1,k}$ to the corresponding registers. 
Here $P_2^{R,P_1,k}$ is obtained from $P_2$ through a controlled reprogramming procedure that 
depends on the queries made to $R$ and $P_1$: $P_2$ is updated on the fly via controlled unitaries, 
without ever measuring $\A$’s state (details will be specified later). 
At the end of the interaction, $\A$’s view is equivalent to having oracle access to $(R, P_1, P_2^{R,P_1,k})$.

 \item[$\Hyb_2$:] \textbf{Ideal World.}  
The simulator $\mathcal{S}$ samples $R, P_1, P_2 \leftarrow \Perms_n$ and $k \leftarrow D$, 
and prepares the joint state $\ket{\Psi_0}$ after receiving the query state $\ket{\Phi_0}$ from $\A$. 
The simulator then replaces the encryption oracle $E_k$ with an independent random permutation $R$, and answers $\A$’s queries by applying $R$, $P_1$, and $P_2$ to the corresponding registers. At the end of the interaction, $\A$’s view is equivalent to having oracle access to $(R, P_1, P_2)$.

\end{enumerate}

Note that in all of the above hybrids, $\mathcal{S}$ returns only the registers $X, A, C, Y, B, D$ to $\mathcal{A}$. From the above discussion, it is straightforward to describe the final states of 
$\Hyb_0$ and $\Hyb_2$. In $\Hyb_0$, the adversary’s interaction yields
\begin{align} \label{eqn:H0}
\ket{\Psi^{(0)}} \;=\;
\sum_{\vec{a},\vec{c}} \alpha_{\vec{a}} \beta_{\vec{c}}\;
\ket{\vec{x},\vec{a},\vec{c}}_{XAC}
\;\otimes\;
\ket{E_k(\vec{x}),\,P_1(\vec{a}),\,P_2(\vec{c})}_{YBD}
\;\otimes\;
\ket{k}_{K}\,\ket{P_1^{\pm}}_{P_1^{\mathrm{reg}}}\,\ket{P_2^{\pm}}_{P_2^{\mathrm{reg}}}.
\end{align}
Similarly, in $\Hyb_2$ we obtain
\begin{align} \label{eqn:H2}
\ket{\Psi^{(2)}} \;=\;
\sum_{\vec{a},\vec{c}} \alpha_{\vec{a}} \beta_{\vec{c}}\;
\ket{\vec{x},\vec{a},\vec{c}}_{XAC}
\;\otimes\;
\ket{R(\vec{x}),\,P_1(\vec{a}),\,P_2(\vec{c})}_{YBD}
\;\otimes\;
\ket{k}_{K}\,\ket{P_1^{\pm}}_{P_1^{\mathrm{reg}}}\,\ket{P_2^{\pm}}_{P_2^{\mathrm{reg}}}.    
\end{align}

The non-trivial case is $\Hyb_1$, where $E_k$ is replaced by a random permutation 
$R$ and $P_2$ is controlled reprogrammed into $P_2^{R,P_1,k}$ to maintain consistency. 
We now describe in detail how to simulate $\Hyb_1$.

\subsection*{Our simulator for $\Hyb_1$.}
\noindent
\emph{(i) Answer $R$ classically and apply the unitary corresponding to $P_1$.}
After answering $q_E$ classical queries with $R$ and performing the $q_1$ calls to $P_1$, the joint (pure) state (before touching $P_2$) can be written compactly as
\[
\ket{\Psi_1^{(1)}}\;=\;
\sum_{\vec{a},\vec{c}} \alpha_{\vec{a}} \beta_{\vec{c}}\;
\ket{\vec{x},\vec{a},\vec{c}}_{XAC}
\;\otimes\;
\ket{R(\vec{x}),P_1(\vec{a}),0^{q_{P_2}}}_{YBD}\;\otimes\;
\ket{k}_{K}\,\ket{P_1^{\pm}}_{P_1^{\mathrm{reg}}}\,\ket{P_2^{\pm}}_{P_2^{\mathrm{reg}}},
\]

\medskip
\noindent
\emph{(ii) Controlled on $R$ and $P_1$ queries, coherently compute flags and targets for reprogramming.}
For each classical query $x_r$ with $r\in [q_E]$, define
\[
g_r:\;(\bool^n)^{q_{P_1}}\to \bool,\qquad
g_r(\vec a)\;=\;\bigvee_{i=1}^{q_{P_1}}\mathbf{1}\big[\,a_i = x_r\oplus k_0\,\big].
\]
The simulator computes $g_r(\vec a)$ coherently into fresh ancilla bits $G_r$ via a reversible circuit
\(
\ket{\vec a}_A\ket{0}_{G_r}\mapsto \ket{\vec a}_A\ket{g_r(\vec a)}_{G_r}.
\)
Let $G=(G_1,\dots,G_{q_E})$. The resulting state is
\begin{align*}
\ket{\Psi_2^{(1)}}\;=\;
\sum_{\vec{a},\vec{c}} \alpha_{\vec{a}} \beta_{\vec{c}}\;
\ket{\vec{x},\vec{a},\vec{c}}_{XAC}
\;\otimes\;
\ket{R(\vec{x}),\,P_1(\vec{a}),\,0^{q_{P_2}}}_{YBD}
&\;\otimes\;
\ket{k}_{K}\,\ket{P_1^{\pm}}_{P_1^{\mathrm{reg}}}\,\ket{P_2^{\pm}}_{P_2^{\mathrm{reg}}}\\
&\;\otimes\;\bigotimes_{r=1}^{q_E}\ket{g_r(\vec a)}_{G_r}\,.
\end{align*}

Next, for each $r\in[q_E]$, the simulator computes the \emph{deterministic} values
\[
X_r \;:=\; P_1(x_r\oplus k_0)\oplus k_1,\qquad 
Y_r \;:=\; P_2^{-1}\big(R(x_r)\oplus k_2\big),
\]
and records them in fresh ancilla registers $U_r,V_r$ (where $U=(U_1,\dots,U_{q_E})$, $V=(V_1,\dots,V_{q_E})$). Since these registers depend only on the simulator’s randomness $(k,P_1,P_2,R)$ and $x_r$, they can be computed reversibly without touching $\vec a,\vec c$:
\begin{align*}
\ket{\Psi_3^{(1)}}\;=\;
\sum_{\vec{a},\vec{c}} \alpha_{\vec{a}} \beta_{\vec{c}}\;
\ket{\vec{x},\vec{a},\vec{c}}_{XAC}
\;\otimes\;
\ket{R(\vec{x}),\,P_1(\vec{a}),\,0^{q_{P_2}}}_{YBD}
&\;\otimes\;
\ket{k}_{K}\,\ket{P_1}_{P_1^{\mathrm{reg}}}\,\ket{P_2}_{P_2^{\mathrm{reg}}}\\
&\;\otimes\; \bigotimes_{r=1}^{q_E}\ket{g_r(\vec a)}_{G_r}\ket{X_r}_{U_r}\ket{Y_r}_{V_r}.
\end{align*}

\begin{remark}
(i) Computing $g_r$ coherently entangles $G_r$ with $A$; $\mathcal{S}$ will uncompute $G$ after applying the controlled reprogramming so no information is leaked. 
(ii) The evaluations of $X_r,Y_r$ are classical/deterministic given internal registers $(K,P_1^{\mathrm{reg}},P_2^{\mathrm{reg}})$ and $R$, and are implemented by reversible circuits; they do not disturb $A,C$.    
\end{remark}

Next, the simulator applies the controlled reprogramming unitary on the $P_2^{\mathrm{reg}}$ register
\[
U_{\mathrm{reprog}}
=\prod_{r=1}^{q_E}
\Big(\ket{0}\!\bra{0}_{G_r}\otimes I
\;+\;\ket{1}\!\bra{1}_{G_r}\otimes \Swap_{P_2^{\mathrm{reg}}}(X_r,Y_r)\Big),
\]
which, for each $r$, swaps the entries of the permutation register $P_2^{\mathrm{reg}}$ at locations $X_r,Y_r$ whenever $G_r=1$. 
Here $\Swap_{P_2^{\mathrm{reg}}}(X,Y)$ acts as
\[
\Swap_{P_2^{\mathrm{reg}}}(X,Y)\ket{x}_{P_2^{\mathrm{reg}},X}\ket{y}_{P_2^{\mathrm{reg}},Y}
=\ket{y}_{P_2^{\mathrm{reg}},X}\ket{x}_{P_2^{\mathrm{reg}},Y}.
\]

For each $\vec a$, define the set of active reprogramming points
\[
S(\vec a) := \{\,r\in[q_E]\ :\ g_r(\vec a)=1\,\}.
\]
Then after applying $U_{\mathrm{reprog}}$, the state becomes 
\begin{align*}
\ket{\Psi_4^{(1)}}\;=\;
\sum_{\vec{a},\vec{c}} \alpha_{\vec{a}} \beta_{\vec{c}}\;
\ket{\vec{x},\vec{a},\vec{c}}_{XAC}
\;\otimes\;
\ket{R(\vec{x}),\,P_1(\vec{a}),\,0^{q_{P_2}}}_{YBD}
&\;\otimes\;
\ket{k}_{K}\,\ket{P_1}_{P_1^{\mathrm{reg}}}\,\ket{P_2^{R,P_1,k}}_{P_2^{\mathrm{reg}}}\\
&\;\otimes\; \bigotimes_{r=1}^{q_E}\ket{g_r(\vec a)}_{G_r}\ket{X_r}_{U_r}\ket{Y_r}_{V_r}\,,
\end{align*}
where $\ket{P_2^{R,P_1,k}}_{P_2^{\mathrm{reg}}}$ denotes the register $P_2^{\mathrm{reg}}$ after all swaps at indices in $S(\vec a)$ have been applied, i.e.,
\[
P_2^{R,P_1,k} \;:=\; P_2 \circ \Big(\prod_{r\in S(\vec a)} \switch_{X_r,Y_r}\Big),
\]
where $P \circ \switch$ is defined by the following equation. 
\begin{align} \label{eqn:swap} 
(P\circ \switch_{a,b})(z) \;=\;
\begin{cases}
P(b), & \text{if } z=a, \\
P(a), & \text{if } z=b, \\
P(z), & \text{otherwise}.
\end{cases}   
\end{align}

\emph{(iii) Answer the $q_2$ calls to the reprogrammed $P_2^{R,P_1,k}$ and uncompute the ancilla registers.} The resulting state is
\begin{align} \label{eqn:H1}
\ket{\Psi_5^{(1)}}&\;=\;
\sum_{\vec{a},\vec{c}} \alpha_{\vec{a}} \beta_{\vec{c}}\;
\ket{\vec{x},\vec{a},\vec{c}}_{XAC}
\;\otimes\;
\ket{R(\vec{x}),\,P_1(\vec{a}),\,P_2^{R,P_1,k}(\vec{c})}_{YBD}
\;\otimes\;
\ket{k}_{K}\,\ket{P_1}_{P_1^{\mathrm{reg}}}\,\ket{P_2^{R,P_1,k}}_{P_2^{\mathrm{reg}}} \nonumber\\ 
&=\ket{\Psi^{(1)}}.
\end{align}

\noindent Next, we bound the distinguishing probability for adjacent hybrids. Let $\Pr[\A(\Hyb_i)=1]$ be the probability that $\A$ outputs 1 in experiment $\Hyb_i$.

\begin{lemma}\label{lem:H1H2}
Let \textbf{$\Hyb_1$} and \textbf{$\Hyb_2$} be defined as above. The distinguishing advantage for an adversary $\A$ between experiments $\Hyb_1$ and $\Hyb_2$ is bounded by:
\begin{equation*}
\mathbb{E}\left|\PrA{1} - \PrA{2}\right| \leq 2q_{P_2}\,q_{P_1}\,\frac{\sqrt{\,q_E}}{2^n}.
\end{equation*}
\end{lemma}
\begin{proof}
We aim to bound the trace distance between $\ket{\Psi^{(1)}}$ and $\ket{\Psi^{(2)}}$, i.e., 

\[
\Tr(\Psi^{(1)}, \Psi^{(2)})=\frac{1}{2}\| \ket{\Psi^{(1)}}\bra{\Psi^{(1)}}-\ket{\Psi^{(2)}}\bra{\Psi^{(2)}}\|_1=\sqrt{1-|\bra{\Psi^{(1)}}\Psi^{(2)}\rangle|^2}.
\]

Recall \expref{Equation}{eqn:H2} and \ref{eqn:H1}. 
The only difference between $\ket{\Psi^{(1)}}$ and $\ket{\Psi^{(2)}}$ 
is that in $\ket{\Psi^{(1)}}$ the quantum queries are answered by $(P_1,P_2)$, 
whereas in $\ket{\Psi^{(2)}}$ they are answered by $(P_1,P_2^{R,P_1,k})$. 
Let $X_E=\{x_1,\dots,x_{q_E}\}$. 
The reprogrammed permutation $P_2^{R,P_1,k}$ is defined by controlled reprogramming $P_2$ so that, 
for all $x \in X_E$, 
\[
P_2^{R,P_1,k}\big(P_1(x \oplus k_0)\oplus k_1\big) \;=\; R(x) \oplus k_2 \,.
\]
To determine where to reprogram $P_2$, we need to know the values of 
$P_1(x \oplus k_0)$. Since the adversary’s queries to $P_1$ are quantum, 
each of the $q_{P_1}$ queries can, in principle, have support on any of the 
classical points $x \oplus k_0$ for $x \in X_E$. For each $P_1$-query 
indexed by $i \in [q_{P_1}]$, we define the set of classical inputs on which 
the query has nonzero amplitude:
\[
  T_i \;:=\; \{\, x \in X_E : 
  \text{the $i$-th $P_1$-query has nonzero amplitude on $x\oplus k_0$} \,\}.
\]

These points are then mapped through the permutation $P_1$ to define the 
corresponding reprogramming set in the domain of $P_2$:
\[
  \Delta_i \;:=\; P_1(T_i)\oplus k_1 \;\subseteq\; \{0,1\}^n.
\]
Because the adversary is non-adaptive across oracles and $R, P_1, P_2$, and $k$ are independent of each other, $T_i$ and $\Delta_i$ are random sets. Thus, by Lemma~\ref{lem:avg}, for each $i \in [q_{P_1}]$ the expected 
number of marked inputs in the $i$-th query is
\[
  \mathbb{E}[\,|T_i|\,] 
  \;=\; \sum_{x \in X_E} \mathbb{E}\bigl[p^{(P_1)}_i(x)\bigr] 
  \;=\; \frac{q_E}{2^n}.
\]
Since $P_1$ is a permutation, $P_1(\cdot)\oplus k_1$ is also a permutation and hence 
$|\Delta_i| = |T_i|$. In particular,
\[
  \mathbb{E}[\,|\Delta_i|\,] = \frac{q_E}{2^n}.
\]

For each query index $i$, we must reprogram $P_2$ on the image of $\Delta_i$. 
To enforce consistency, we modify $P_2$ to a new permutation $P_2'$ that 
satisfies the desired mapping on $\Delta_i$ while still being a permutation. 
This can be done by implementing a series of disjoint swaps: for each 
$t\in \Delta_i$, we swap $P_2(t)$ with the desired value $v_t$. 
Each such operation changes $P_2$ only on $t$ and 
$P_2^{-1}(v_t)$, so the set of altered points $S_i := \{u : P_2^{R,P_1,k}(u)\neq P_2(u)\}$ 
has size at most twice that of the reprogrammed set, i.e., $|S_i|\le 2|\Delta_i|$. 
Finally, let
\[
S \;:=\; \bigcup_{i=1}^{q_{P_1}} S_i
\]
denote the total set of reprogrammed points in $P_2$ across all $q_{P_1}$ queries.

\paragraph{Distance bound.}
The total distinguishability arises from the $q_{P_2}$ quantum queries to $P_2$
potentially hitting the reprogrammed set $S$. For the $j$-th query
($j\in[q_{P_2}]$), the hybrid/triangle inequality (Lemma~\ref{lem:triangle})
implies that the per-query contribution is at most
$2\sqrt{\,2p^{(P_2)}_j(S)\,}$. Summing over queries gives
\[
  \tfrac{1}{2}\bigl\|
    \ket{\Psi^{(1)}}\!\bra{\Psi^{(1)}} - \ket{\Psi^{(2)}}\!\bra{\Psi^{(2)}}
  \bigr\|_1
  \;\le\; \sqrt{2}\sum_{j=1}^{q_{P_2}} \sqrt{\,p^{(P_2)}_j(S)\,},
\]
where $S=\bigcup_{i=1}^{q_{P_1}} S_i$. We then apply the union bound
$p^{(P_2)}_j(\!\bigcup_i S_i)\le \sum_i p^{(P_2)}_j(S_i)$ and the concavity
$\sqrt{\sum_i a_i}\le \sum_i \sqrt{a_i}$ to obtain
\[
  \tfrac{1}{2}\bigl\|
    \ket{\Psi^{(1)}}\!\bra{\Psi^{(1)}} - \ket{\Psi^{(2)}}\!\bra{\Psi^{(2)}}
  \bigr\|_1
  \;\le\; \sqrt{2} \sum_{j=1}^{q_{P_2}} \sqrt{\,p^{(P_2)}_j\!\bigl(\textstyle\bigcup_{i} S_i\bigr)\,}
  \;\le\;\sqrt{2} \sum_{j=1}^{q_{P_2}} \sum_{i=1}^{q_{P_1}} \sqrt{\,p^{(P_2)}_j(S_i)\,}.
\]
We now take the expectation of an individual term for a fixed $P_2$-query $j$ and reprogram set $S_i$. Using Lemma~\ref{lem:avg} and our bounds on the size of $S_i$:
\begin{equation} \label{eqn:S}
\mathbb{E}\bigl[p^{(P_2)}_j(S_i)\bigr]\ =\ \frac{\mathbb{E}[|S_i|]}{2^n}\ \le\ \frac{2\,\mathbb{E}[|\Delta_i|]}{2^n}\ =\ \frac{2\,q_E}{2^{2n}},   
\end{equation}

By Jensen's inequality, $\mathbb{E}\bigl[\sqrt{p^{(P_2)}_j(S_i)}\bigr] \le \sqrt{\mathbb{E}\bigl[p^{(P_2)}_j(S_i)\bigr]} \le \sqrt{2\,q_E}/2^n$. Summing this bound over all $j$ and $k$ gives the total expected trace distance:
\[
  \mathbb{E}\bigl[\ \Tr(\Psi^{(1)}, \Psi^{(2)})\bigr]\ \le\ \sqrt{2}\sum_{j=1}^{q_{P_2}}\sum_{j=1}^{q_{P_1}} \frac{\sqrt{2\,q_E}}{2^n} = 2q_{P_2}\,q_{P_1}\,\frac{\sqrt{\,q_E}}{2^n}.
\]

\end{proof}

\begin{lemma}\label{lem:H0H1}
Let \textbf{$\Hyb_0$} and \textbf{$\Hyb_1$} be defined as above. The distinguishing advantage for an adversary $\A$ between experiments $\Hyb_0$ and $\Hyb_1$ is bounded by:
\begin{equation*}
\mathbb{E}\left|\PrA{0} - \PrA{1}\right| \leq P^{\sf coll}\cdot
        2q_{P_2}\,q_{P_1}\,\frac{\sqrt{\,q_E}}{2^n},
\end{equation*}
where $P^{\sf coll} \;=\; \min\!\left(1, \;\frac{q_E^{2}}{2^n}\right)$.
\end{lemma}
\begin{proof}

In $\Hyb_0$, all $q_E$ classical queries are answered by the cipher $E_k$, 
while in $\Hyb_1$ they are answered by an independent random permutation $R$. 
The adversary $\A$ has quantum access to both $P_1$ and $P_2$ in both hybrids. 
A potential divergence arises when the input to $P_2$ coincides with a point of 
the form $x = b_i \oplus k_1$ that is correlated with some classical query 
$(x_i,y_i)$ to $E$. 

Concretely, suppose $\A$ makes a quantum query to $P_1$ with superposition 
coefficients $\alpha_{a}$ over inputs $a$. If there exists an index $i$ such that 
$a_i = x \oplus k_0$ for some $x \in X_E$ with $\alpha_{a_i} \neq 0$, then the 
corresponding value $b_i = P_1(a_i)$ also appears with non-zero amplitude. 
This in turn induces a potential query to $P_2$ at 
\[
   c_j = b_i \oplus k_1,
\]
again with non-zero amplitude. 

In $\Hyb_0$, such an input is consistent with the encryption relation:
\[
   P_2(c_j) = P_2(b_i \oplus k_1) = y_i \oplus k_2,
\]
since $y_i = E_k(x_i)$. In contrast, in $\Hyb_1$ we have $y_i = R(x_i)$ 
uniformly random, so the relation with $P_2$ is broken. 

To restore consistency, we reprogram $P_2$ on these inputs: for all 
$x \in X_E$,
\[
   P_2^{R,P_1,k}\bigl(P_1(x \oplus k_0)\oplus k_1\bigr)
   = R(x)\oplus k_2.
\]
This reprogramming is triggered only on inputs where the adversary’s quantum query has non-zero amplitude, ensuring that from $\A$’s perspective the distributions in $\Hyb_0$ and $\Hyb_1$ are identical. 

Therefore, the two hybrids are \emph{exactly identical}
as quantum processes \emph{except} on the event of a \emph{swap collision}, namely when one of the desired swap inputs coincides with one of the preimages to be swapped. Recall that for each $r\in[q_E]$,
\[
X_r \;:=\; P_1(x_r\oplus k_0)\oplus k_1,\qquad 
Y_r \;:=\; P_2^{-1}\big(R(x_r)\oplus k_2\big).
\]
Let
\[
  X \;:=\; \{X_1,\ldots,X_{q_E}\},\qquad
  Y \;:=\; \{Y_1,\ldots,Y_{q_E}\}.
\]
If $X \cap Y = \emptyset$, then all transpositions $(X_r\, Y_r)$ are disjoint,
so the (product) swap exactly implements the reprogramming and $\Hyb_0$ and
$\Hyb_1$ are identical from $\A$'s perspective. Conversely, if $X \cap Y \neq
\emptyset$, some transpositions overlap, the swap procedure misfunctions on those
points, and the two hybrids may differ (and only on basis inputs that place
nonzero amplitude on such collision points).

To analyze the effect of swap collisions, recall that
\[
  X = \{X_1,\ldots,X_{q_E}\},\qquad 
  Y = \{Y_1,\ldots,Y_{q_E}\},
\]
where $X_r = P_1(x_r\oplus k_0)\oplus k_1$ and 
$Y_r = P_2^{-1}(R(x_r)\oplus k_2)$. 
By construction, $X$ depends only on $(k_0,k_1,P_1)$, while $Y$ depends only on
$(k_2,R,P_2)$. Since these keys and permutations are chosen independently, both $X$ and $Y$ are uniformly distributed $q_E$-subsets of $\{0,1\}^n$, and they are independent.
For $q_E \ll 2^n$, the probability of swap collision can by approximated by
\[
  \Pr[X \cap Y = \emptyset] 
  = \prod_{i=0}^{q_E-1} \frac{2^n - q_E - i}{2^n - i}
  = \prod_{i=0}^{q_E-1} \Bigl(1 - \frac{q_E}{2^n - i}\Bigr)
  \approx \Bigl(1 - \frac{q_E}{2^n}\Bigr)^{q_E}
  \approx e^{-q_E^2/2^n}.
\]
Consequently,
\[
  \Pr[X \cap Y \neq \emptyset] 
  = 1 - \Pr[X \cap Y = \emptyset]
  \approx 1 - e^{-q_E^2/2^n}.
\]
Using the inequality $e^{x} \ge 1+x$ with $x=-q_E^2/2^n$, we obtain
$1 - e^{-q_E^2/2^n} \;\le\; \frac{q_E^2}{2^n}$, which implies
\[
\Pr[X \cap Y \neq \emptyset] \;\le\; \frac{q_E^2}{2^n}.
\]
Accordingly, we define the collision probability as
\[
P^{\sf coll} \;=\; \min\!\left(1, \;\frac{q_E^{2}}{2^n}\right).
\]

When no collision occurs, the two hybrids are \emph{identical} since
the reprogramming is implemented exactly. If a collision occurs, the hybrids may differ, but only on inputs lying in the collision set 
$\mathsf{Coll} := X \cap Y$. Recall that 
$S=\bigcup_{i=1}^{q_{P_1}} S_i$ is the controlled reprogram set covering all
possible reprogramming locations. Clearly $\mathsf{Coll} \subseteq S$: every 
collision must occur at some point in $S$, but not every point in $S$ is 
necessarily involved in a collision. Hence, for each $j$,
\[
  p^{(P_2)}_j(\mathsf{Coll}) \;\le\; p^{(P_2)}_j(S).
\]
Conditioned on $\mathsf{Coll}$ we have
\[
 \tfrac{1}{2}\bigl\|
    \ket{\Psi^{(0)}}\!\bra{\Psi^{(0)}} - \ket{\Psi^{(1)}}\!\bra{\Psi^{(1)}}
  \bigr\|_1 \;\le\; \sqrt{2}\sum_{j=1}^{q_{P_2}} \sqrt{\,p^{(P_2)}_j(\mathsf{Coll})\,}
  \;\le\; \sqrt{2}\sum_{j=1}^{q_{P_2}} \sqrt{\,p^{(P_2)}_j(S)\,}.
\]

When there is no collision, we have $\tfrac{1}{2}\bigl\| \ket{\Psi^{(0)}}\!\bra{\Psi^{(0)}} -\ket{\Psi^{(1)}}\!\bra{\Psi^{(1)}} \bigr\|_1=0$. By applying \expref{Equation}{eqn:S}, we have 
\[
  \mathbb{E}\left|\PrA{0} - \PrA{1}\right| \leq P^{\sf coll}\cdot \mathbb{E}[ \Bigl(\sqrt{2}\sum_{j=1}^{q_{P_2}} \sqrt{\,p^{(P_2)}_j(S)\,}\Bigr)]
  \leq P^{\sf coll}\cdot 2q_{P_2}\,q_{P_1}\,\frac{\sqrt{\,q_E}}{2^n},
\]
where $P^{\sf coll} \;=\; \min\!\left(1, \;\frac{q_E^{2}}{2^n}\right)$.
\end{proof}
Applying the triangle inequality to combine \expref{Lemma}{lem:H0H1} and \expref{Lemma}{lem:H1H2}, 
and noting that $P^{\sf coll}\leq 1$, we obtain the final bound.

\end{proof}

\begin{remark}\label{lbl:rem_mixed}
For the proof of \expref{Theorem}{thm:fullQ1} we assumed that the adversary’s initial state is pure without the presence of an auxiliary register. However, in general, the adversary’s pre-query state can be mixed and even entangled with an arbitrary reference system~$R$. We now explain why the theorem continues to hold in this more general setting.

(i) The distinguishing advantage of the adversary between hybrids $\mathcal{H}_0$ and $\mathcal{H}_1$ depends solely on its queries to the cipher and is independent of how the initial state is prepared. Consequently, the advantage is the same whether the adversary begins in a mixed state entangled with~$R$ or in a pure state with no auxiliary registers.

(ii) Distinguishing $\mathcal{H}_1$ from $\mathcal{H}_2$ reduces to distinguishing the post-query states when the oracle answers with the original permutation versus the reprogrammed one; this is bounded using \expref{Lemma}{lem:single} together with \expref{Lemma}{lem:triangle}. For a pure initial state, we instead apply \expref{Lemma}{lem:tr-ms} in combination with \expref{Lemma}{lem:triangle}, which yields the same upper bound; hence the resulting bound on the distinguishing advantage is identical in both cases.

Combining (i) and (ii), the upper bound on the adversary’s distinguishing advantage is unchanged even when its initial state is an arbitrary mixed state entangled with a reference system. Hence, in the proof we analyze the adversary’s state as the canonical pure state without loss of generality.

\end{remark}

\subsection{Extending to t-KAC}
Here, we show that \expref{Theorem}{thm:fullQ1} can be extended to the $t$-KAC case.

\begin{theorem}[Security of $t$-KAC, Q1 model]\label{thm:fullQ1-t}
     Let	$D$ be a distribution over $k=(k_0, \ldots, k_{t})$ 
such that the marginal distributions of $k_0, \ldots, k_{t}$ are each uniform. Let $\A$ be a non-adaptive adversary that makes $q_E$ classical queries to the cipher and $q_{P_1},\ldots,q_{P_{t}}$ quantum queries to $P_1,\ldots,P_t$. Let $\mathbf{Adv}_{\text{$2$-KAC,Q1}}(\A)$ denote the qPRP advantage of $\A$:

   \begin{equation*}
        \mathbf{Adv}_{\text{$t$-KAC,Q1}}(\A) :=
     \mathbb{E}_{\substack{k \leftarrow D \\ R, P_1,\dots,P_t \leftarrow \Perms_n }}\left|\Pr \left[\A^{E_k, P_1,\dots,P_t}(1^n) = 1\right]
    - \Pr \left[\A^{R, P_1,\dots,P_t}(1^n) = 1\right]\right|.
   \end{equation*}

    We obtain $    \mathbf{Adv}_{\text{$t$-KAC,Q1}}(\A) \leq 4q_{P_1}\cdots q_{P_t}\frac{\sqrt{\,q_E}}{2^{tn/2}} $.
     \end{theorem}

\begin{proof}
The proof follows the same strategy as \expref{Theorem}{thm:fullQ1}. 
We again restrict to a non-adaptive adversary that commits to all its queries 
(to $E_k$, $P_1,\dots,P_t$) before interacting with the oracles. 
That is, the initial state is
\[
\ket{\Psi_0}
=\!\sum_{\vec a^{(1)},\ldots,\vec a^{(t)}}
\alpha_{\vec a^{(1)},\ldots,\vec a^{(t)}}
\;\ket{\vec x}_{X}
\!\Big(\bigotimes_{i=1}^t \ket{\vec a^{(i)}}_{A^{(i)}}\Big)
\!\Big(\ket{0^{q_E}}_{Y} \bigotimes_{i=1}^t \ket{0^{q_{P_i}}}_{B^{(i)}}\Big)
\otimes\ket{k}_{K}
\!\Big(\bigotimes_{i=1}^t \ket{P_i}_{(P_i)^{\mathrm{reg}}}\Big),
\]

where $\vec{a}^{(i)}=(a_1^{(i)},\dots, a_{q_{P_i}}^{(i)})$ denotes a superposition of quantum queries $(a^{(i)}_1,\dots,a^{(i)}_{q_{P_i}})$ to $P_i$ with amplitudes $\alpha_{\vec{a^{(i)}}}$. Define the hybrids:
\begin{itemize}
  \item[$\Hyb_0$:] $(E_k, P_1,\dots,P_t)$
  \item[$\Hyb_1$:] $(R, P_1,\dots,P_{t-1}, P^{R,P_1,\dots,P_{t-1}}_t)$
  \item[$\Hyb_2$:] $(R, P_1,\dots,P_t)$.
\end{itemize}

The permutation $P^{R,P_1,\dots,P_{t-1}}_t$ is obtained by controlled reprogramming 
$P_t$ on the fly, exactly as in the 2-KAC case. Whenever a superposition query passes 
through $P_1,\dots,P_{t-1}$ and reaches a point corresponding to some $x\in X_E$, we 
reprogram $P_t$ so that
\[
   P^{R,P_1,\dots,P_{t-1}}_t\!\bigl(P_{t-1}(\cdots P_1(x\oplus k_0)\cdots)\oplus k_{t-1}\bigr)
   \;=\; R(x)\oplus k_t.
\]
This ensures that $\Hyb_0$ and $\Hyb_1$ remain consistent from the adversary’s view. 
The subsequent trace-distance analysis carries over verbatim from the 2-KAC case, 
yielding the desired bound.
\end{proof}

Our theorem proves that $\A$ is required to make $q^2_{P_1} \cdots q^2_{P_t} \cdot q_E= \Omega (2^{tn})$ to achieve constant distinguishing advantage. 
When $t=1$, approximately $2^{n/3}$ classical and quantum queries are necessary for a successful attack, which matches the result in \cite{alagic2022post}. When $t\geq 2$,  approximately $2^{\frac{tn}{(2t+1)}}$ queries are necessary for a non-adaptive adversary to attack $t$-KAC in the Q1 model.

Our proofs bound the expected value of the adversary’s distinguishing advantage, taken over the randomness of keys and permutations.
The following remark shows that this suffices, as an expectation bound already implies that the real advantage is small for all but a negligible fraction of choices.

\begin{remark}
Let $\Adv_{\A}$ denote the (random) distinguishing advantage of a fixed adversary $\A$,
where the randomness is over the choice of keys, permutations, and $\A$’s internal coins.
By Markov's inequality, for any $\gamma \in (0,1]$,
\[
  \Pr\!\big[\Adv_{\A} \ge \gamma\big]
  \ \le\ \frac{\E[\Adv_{\A}]}{\gamma}.
\]
In particular, for any polynomial $p(n)$,
\[
  \Pr\!\big[\Adv_{\A} \ge 1/p(n)\big]
  \ \le\ p(n)\cdot \E[\Adv_{\A}].
\]
Hence, if $\E[\Adv_{\A}] = \negl(n)$, then $\Pr[\Adv_{\A} \ge 1/p(n)] = \negl(n)$
for every polynomial $p(n)$.
That is, for all but a negligible fraction of choices of keys and permutations,
the actual advantage $\Adv_{\A}$ is at most $1/p(n)$.
Consequently, an expectation bound immediately implies that, except with negligible probability,
the realized advantage is not far from its expected value.
\end{remark}

\section{Security of KAC in Q2 model}
\label{sec:q2lb}
Next, we analyze the Q2 security of $t$-KAC, where the adversary is given quantum access to all oracles. To do this, we introduce a simple oracle lifting argument to reduce the Q1 lower bound to a Q2 lower bound for any cipher, from which we then derive the Q2 security of $t$-KAC.

Before presenting our results, we introduce some terminology. Let $F:\{0,1\}^n\to\{0,1\}^n$ and $\mathcal{G}=\{G_1,\dots,G_m\}$. An algorithm may have
(i) \emph{classical access} to $F$ (computational basis queries only), counted by $q_F^{\mathrm{class}}$,
or (ii) \emph{quantum superposition access} to $F$ via the unitary $U_F:\ket{x,y}\mapsto\ket{x,y\oplus F(x)}$.
For each $G_i$ we count quantum queries by $q_i$.

\begin{lemma}[Table-simulation]
\label{lem:table-sim}
Given the full truth table $T_F=\{(x,F(x)) : x\in\{0,1\}^n\}$, one can implement $U_F$ by a reversible circuit without making oracle queries to $F$.
Equivalently, an algorithm with $2^n$ classical queries to $F$ (to learn $T_F$) can thereafter answer arbitrary quantum queries to $F$ at zero additional \emph{query} cost.
\end{lemma}

\begin{proposition}[Necessity of $2^n$ classical queries for exact simulation]
\label{prop:learn-2n}
Any algorithm that, for every function $F:\{0,1\}^n\to\{0,1\}^n$, outputs a circuit that agrees exactly with $U_F$ on all inputs must make $2^n$ classical queries to $F$ in the worst case.
\end{proposition}

\begin{proof}
Suppose the algorithm makes at most $2^n-1$ classical queries. Then there exists $x^\star$ it never queries. Construct $F,F'$ identical everywhere except at $x^\star$. The transcript is identical on both, hence the algorithm outputs the same circuit, which cannot implement both $U_F$ and $U_{F'}$. Thus $2^n$ classical queries are necessary.
\end{proof}

\begin{theorem}[Oracle Lifting]
\label{thm:oracle-lifting}
Fix \emph{Real/Ideal} distributions over oracles $(F,\mathcal{G})$ and let $\Adv(\cdot)$ be the distinguishing advantage.
Suppose for every algorithm $\mathcal{A}$ with classical access to $F$ and quantum access to $\mathcal{G}$ that makes $q_F^{\mathrm{class}}$ classical queries to $F$ and $q_i$ quantum queries to $G_i$, we have
\begin{equation}
\label{eq:q1-bound}
\Adv(\mathcal{A}) \le \varepsilon
\quad\text{whenever}\quad
\Phi\!\bigl(q_F^{\mathrm{class}},\,q_1,\dots,q_m\bigr)\;\le\;B.
\end{equation}
Then any algorithm $\mathcal{B}$ with quantum access to all of $F,G_1,\dots,G_m$, making $q_F^{\mathrm{quantum}}$ queries to $F$ and $q_i$ to $G_i$, satisfies
\begin{equation}
\label{eq:q2-instantiation}
\Adv(\mathcal{B}) \le \varepsilon
\quad\text{whenever}\quad
\Phi\!\bigl(\mathbf{2^n},\,q_1,\dots,q_m\bigr)\;\le\;B.
\end{equation}
\end{theorem}

\begin{proof}
Let $\mathcal{B}$ be any quantum-access algorithm. Construct $\mathcal{S}$ that first learns $T_F$ using $2^n$ classical queries, then simulates $U_F$ by \expref{Lemma}{lem:table-sim}, and otherwise forwards $\mathcal{B}$’s quantum queries to the $G_i$. This guarantees $\Adv(\mathcal{S})=\Adv(\mathcal{B})$. Moreover, $\mathcal{S}$ makes $q_F^{\mathrm{class}}=2^n$ queries to $F$ and the same $q_i$ queries to each $G_i$. Applying \eqref{eq:q1-bound} with $q_F^{\mathrm{class}}=2^n$ yields \eqref{eq:q2-instantiation}.
\end{proof}

\begin{remark}
While \expref{Proposition}{prop:learn-2n} is formulated for arbitrary functions, an analogous requirement applies to any permutation $P:\{0,1\}^n \to \{0,1\}^n$. In this case, bijectivity implies that $2^n - 1$ classical queries are both sufficient and necessary to obtain a circuit that agrees with $U_P$ on all inputs.
\end{remark}

\subsection*{Applications: Security of t-round KAC}
\label{subsec:app-2kac}
In what follows, we show how to derive $Q2$ security for $t$-KAC by applying \expref{Theorem}{thm:oracle-lifting}.

\begin{theorem}[Security of $t$-KAC, Q2 model]\label{cor:kac-q2}
Let	$D$ be a distribution over $k=(k_0, \ldots, k_{t})$ 
such that the marginal distributions of $k_0, \ldots, k_{t}$ are each uniform. Let $\A$ be a non-adaptive adversary that makes $q_{P_0},q_{P_1},\ldots,q_{P_{t}}$ quantum queries to $E_k, P_1,\ldots,P_t$. Let $\mathbf{Adv}_{\text{$t$-KAC,Q2}}(\A)$ denote the qPRP advantage of $\A$:
   \begin{equation*}
        \mathbf{Adv}_{\text{$t$-KAC,Q2}}(\A) :=
     \mathbb{E}_{\substack{k \leftarrow D \\ R, P_1,\dots,P_t \leftarrow \Perms_n }}\left|\Pr \left[\A^{E_k, P_1,\dots,P_t}(1^n) = 1\right]
    - \Pr \left[\A^{R, P_1,\dots,P_t}(1^n) = 1\right]\right|.
   \end{equation*}

    We obtain $    \mathbf{Adv}_{\text{$t$-KAC,Q2}}(\A) \leq \frac{4}{2^{(t-1)n/2}} \,\min_{0\le i\le t}\biggl\{
  \frac{q_{P_0} \cdot q_{P_1} \cdots q_{P_t} }
       {q_{P_i}}
\biggr\}$.
\end{theorem}

\begin{proof}
Let $F:=E_k$ denote the outer keyed permutation and set $(G_1,\ldots,G_t):=(P_1,\ldots,P_t)$ for the internal permutations.
By \expref{Theorem}{thm:fullQ1-t} we obtain
\begin{equation}\label{q1_p0}
\mathbf{Adv}_{\text{$t$-KAC,Q1}}(\A) \leq 4q_{P_1}\cdots q_{P_t}\frac{\sqrt{\,q_{P_0}}}{2^{tn/2}}.
\end{equation}
Applying Theorem~\ref{thm:oracle-lifting} to \eqref{q1_p0} with $q_E^{\sf class}=2^n$ yields
\begin{equation}\label{q2_p0}
\mathbf{Adv}_{\text{$t$-KAC,Q2}}(\A) \leq \frac{4q_{P_1}\cdots q_{P_t}}{2^{(t-1)n/2}}.
\end{equation}
Next, consider $t$ hypothetical adversaries $\text{Q1}_{i}$ for $1\leq i \leq t$, where each has classical access to $P_i$ and quantum access to $E_k, P_1,\ldots,P_{i-1},P_{i+1},\ldots,P_t$. Note that $P_i$ can be viewed as a KAC with permutations $(P_{i-1}^{-1},\ldots,P_1^{-1},E_k,P_t^{-1},\ldots,P_{i+1}^{-1})$ and keys $(k_{i-1},\ldots,k_0,k_t,\ldots,k_i)$. This time take $F:=P_i$ as the outer keyed permutation and $(G_1,\ldots,G_t):=(E_k,P_1,\ldots,P_{i-1},P_{i+1},\ldots, P_t)$ as the internal ones. By another application of \expref{Theorem}{thm:fullQ1-t}, for each $1\leq i\leq t$ we get
\begin{equation}\label{q1_pi}
\mathbf{Adv}_{\text{$t$-KAC,$\text{Q1}_i$}}(\A) \leq 4q_{P_0}\cdots q_{P_{i-1}} q_{P_{i+1}} \cdots q_{P_t}\frac{\sqrt{\,q_{P_i}}}{2^{tn/2}}.
\end{equation}
Applying Theorem~\ref{thm:oracle-lifting} to \eqref{q1_pi} with $q_{P_i}^{\sf class}=2^n$, for all $1\leq i\leq t$ we conclude
\begin{equation}\label{q2_pi}
\mathbf{Adv}_{\text{$t$-KAC,Q2}}(\A) \leq \frac{4q_{P_0}\cdots q_{P_{i-1}} q_{P_{i+1}} \cdots q_{P_t}}{2^{(t-1)n/2}}.
\end{equation}
Combining \eqref{q2_p0} and \eqref{q2_pi} completes the proof.
\end{proof}

Above theorem shows that $\mathcal{A}$ is required to make $\Omega(2^\frac{(t-1)n}{2t})$ quantum queries to each of the $E_k,P_1,\ldots,P_t$. In contrast, $1$-round KAC is completely broken in polynomial quantum time by Simon’s algorithm~\cite{kuwakado2012security}. \expref{Corollary}{cor:kac-q2} for $t$-KAC, with $t\geq 2$ marks the first non-trivial provable Q2 security for any $t$ rounds.

Note that to simulate quantum access to $F$ from classical queries, one must, in general, recover the entire truth table $T_F$. Proposition~\ref{prop:learn-2n} shows that $2^n$ classical
queries are unavoidable in the worst case: omitting even a single input leaves open two
functions $F,F'$ that differ only at that input, yet the corresponding unitaries
$U_F,U_{F'}$ are at diamond distance~$1$. Thus, exact simulation in the query model
requires full table-learning.

\begin{remark}
The $2^n$ cost can be realized in two equivalent ways: either as a \emph{uniform} strategy
that brute-forces all inputs to $F$, or as \emph{non-uniform} advice in which the entire
truth table $T_F$ is given to the algorithm. In both views, the effective query complexity
to $F$ is $2^n$, after which arbitrary quantum queries to $F$ can be answered without
further oracle calls. This overhead is information-theoretically tight; even for
average-case $F$, approximate simulation cannot safely skip inputs, since flipping an
unqueried point changes $U_F$ at unit distance.
\end{remark}

Note that the quality of the Q2 lower bound obtained via Theorem~\ref{thm:oracle-lifting} is governed by two independent effects. First, the $2^n$ overhead for simulating quantum access to $F$ is inherent: \expref{Proposition}{prop:learn-2n} shows that exact table-simulation requires querying every input, and omitting even one leads to unitaries at diamond distance~$1$. 
Second, the Q1 tradeoff bound that we plug in \expref{Theorem}{thm:fullQ1-t} may itself be loose, both because it currently applies only to non-adaptive adversaries and because the precise constants or exponents may not be tight.

%% file: upperbound.tex
\section{Quantum Key-Recovery Attack for KAC}
\label{sec:attack} 

In this section, we present a quantum attack on the KAC cipher within the $Q1$ model. We first revisit the classical attack, highlighting the property we will leverage. 

Given $t$ public permutations $P_1, P_2, \ldots, P_t$ over $n$-bit strings, and $t+1$ independent secret keys $k_0, k_1, \ldots, k_t$ of size $n$, the $t$-round KAC encryption function $E_k:\bool^{(t+1)n}\times \bool^n \rightarrow \bool^n$ for a $t$-KAC on a plaintext $x$ is defined as:
\begin{equation}
   E_k(x) = k_t \oplus P_t(k_{t-1} \oplus P_{t-1}(\ldots P_1(k_0 \oplus x) \ldots)).
\end{equation}

\subsection{Revisiting the Classical Key-Recovery Attack} \label{sec:classical}
The classical attack~\cite{bogdanov2012key} establishes a query complexity of $O(2^{tn/t+1})$. The idea is to sample inputs and look for collisions in generated key candidates. The attack proceeds as follows: First, the adversary \( \mathcal{A} \) randomly selects subsets of the input domains of \( E, P_1, \ldots, P_t \), denoted by \( S_0, \ldots, S_t \). These sets are defined as:

\begin{equation} \label{eqn:sample}
S_i = \{x_{i,1}, \ldots, x_{i,|S_i|}\} \quad \text{for } i=0, \ldots, t.
\end{equation}
with the condition of $|S_0|\cdot |S_1| \cdot \ldots \cdot |S_t|= \beta 2^{tn}$, for some constant \( \beta \geq t+1 \).

\( \mathcal{A} \) queries each of \( E, P_1, \ldots, P_t \) on the inputs taken from the sets \( S_0, \ldots, S_t \), respectively. \( \mathcal{A} \) obtains all query/response pairs $(x_{0,i}, E_k(x_{0,i}))$ for $x_{0,i} \in S_0$, and $(x_{j,i}, P_j(x_{j,i}))$ for $x_{j,i} \in S_j$ where $j \in \{1, \ldots, t\}$. $\mathcal{A}$ then stores these pairs.

Then, for each tuple \( (x_0, \ldots, x_t) \in S_0 \times \cdots \times S_t \), the adversary collect a key candidate \( (\kappa_0, \ldots, \kappa_t) \) via the following sequence of equations:

\begin{equation} \label{eqn:keygen}
\kappa_0 = x_0 \oplus x_1, \quad \kappa_i = P_i(x_i) \oplus x_{i+1} \text{ for } 1 \le i \le t-1, \quad \kappa_t = P_t(x_t) \oplus E_k(x_0).
\end{equation}

Finally, among all the \( \beta 2^{tn} \) generated key candidates, the algorithm returns the one that occurs the most frequently.

To analyze the attack, first observe that since the sets \( S_0, \ldots, S_t \) are chosen randomly, the first \( t \) equations in \expref{Equation}{eqn:keygen} generate each possible key candidate for \( (k_0, \ldots, k_{t-1}) \) approximately \( \beta \) times. Now, suppose the algorithm produces two key candidates that share the same first $t$ components derived from distinct input tuples $(x_0^{(1)}, \ldots, x_t^{(1)})$ and $(x_0^{(2)}, \ldots, x_t^{(2)})$ from $S_0 \times \cdots \times S_t$:
\[
(\kappa_0, \ldots, \kappa_{t-1}, \kappa_t^{(1)}) \quad \text{and} \quad (\kappa_0, \ldots, \kappa_{t-1}, \kappa_t^{(2)}).
\]
Then, the following holds based on the properties of random permutations and keys:
\begin{equation}\label{eqn:lastKeyInequality}
    \Pr[\kappa_t^{(1)} = \kappa_t^{(2)}] =
    \begin{cases}
    1 & \text{if } (\kappa_0, \ldots, \kappa_{t-1}) = (k_0, \ldots, k_{t-1}), \\
    \frac{1}{2^n} & \text{otherwise}.
    \end{cases}
\end{equation}

This implies that $(k_0,\ldots, k_t)$ is the one that appears most frequently among all the generated candidates, with high probability. Now, we can rewrite the classical attack in the following way: 

\vspace{1\baselineskip}
\begin{breakablealgorithm}
\caption{Classical Key-Recovery Attack for $t$-KAC}
\label{att:1}
\begin{enumerate}
    \item  Set $\beta = t+1$. Sample $t+1$ random sets as \expref{Equation}{eqn:sample} conditioning on $$|S_0|\cdot |S_1| \cdot \ldots \cdot |S_t|= \beta2^{tn}.$$
    Set $E:S_0 \rightarrow \bool^n$ and $P_i: S_i \rightarrow \bool^n$ for all $i \in \{1,\ldots,t\}$. 
    \item $\mathcal{A}$ gets classical access to $E,P_1,\ldots, P_t$. $\mathcal{A}$ then queries all possible inputs in $S_0$ for $E$, and all possible inputs in $S_i$ for $P_i$ for $1\leq i \leq t$. $\mathcal{A}$ then stores the query/response pairs.
    \item \label{att:step3}For each input tuple $(x_0,\ldots, x_t) \in S_0\times \ldots\times S_t$ , $\A$ generates a key tuple $(\kappa_0,\ldots,\kappa_t)$ by \expref{Equation}{eqn:keygen}. 
    \item \label{att:step4}$\A$ checks whether there exists a key tuple that is generated by at least $t+1$ different input tuples. It then collects all such key candidates and randomly outputs one of them.
\end{enumerate}
\end{breakablealgorithm}
\vspace{\baselineskip}

The following lemma establishes the optimal query complexity of the algorithm and proves its correctness.
\begin{lemma} \label{lem:successAttack} For any adversary $\A$ executing \expref{Algorithm}{att:1}, the optimal query complexity is $O(2^{tn/(t+1)})$. Furthermore, if $\A$ outputs $\hat{k}$, we have $\Pr[\hat{k} = k] \geq \tfrac{1}{2}$, where the probability is over the random sampling of the sets and the internal randomness of $\A$.
\end{lemma}

\begin{proof}
Suppose an adversary $\A$ makes $q_E$ queries to $E$ and $q_{P_i}$ queries to $P_i$. From \expref{Algorithm}{att:1}, $q_E=|S_0|$ and $q_{P_i}=|S_i|$ for $i \in \{1,\ldots,t\}$. Then $T = q_E+\sum_{i=1}^t q_{P_i}$ is minimized under the constraint $|S_0|\cdot |S_1|\cdot \ldots \cdot|S_t|=\beta 2^{tn}$ when $|S_0|\approx \cdots \approx |S_t|\approx (\beta 2^{tn})^{\tfrac{1}{t+1}}=\beta^{\tfrac{1}{t+1}}2^{\tfrac{tn}{t+1}}$. Hence, the optimal query complexity is $T = O(t\cdot\beta^{\tfrac{1}{t+1}}2^{\frac{tn}{t+1}})=O(2^{\tfrac{tn}{t+1}})$.

For the correctness of the algorithm, consider an arbitrary key candidate $\kappa = (\kappa_0,\ldots,\kappa_t) \neq (k_0,\ldots,k_t)$ that can be generated in \expref{Step}{att:step3}. Let $\kappa^* = (\kappa_0,\ldots,\kappa_{t-1})$ be the first $t$ elements in $\kappa$. \expref{Equation}{eqn:lastKeyInequality} implies that $\kappa^* \neq (k_0,\ldots,k_{t-1})$.

Consider the tuples $(x_0, \ldots, x_t)$ that generate a fixed prefix $\kappa^*$. As mentioned before, the expected number of such tuples for a random $\kappa^*$ is $\beta$. Let those keys be $\kappa^{(j)} = (\kappa^*, \kappa_t^{(j)})$ for $j=1,\ldots, \beta$\footnote{We use the expected number for analysis. By increasing $\beta$, $\A$ can make sure that the correct key is generated at least $t+1$ times, which helps distinguish it from incorrect keys whose generated $\kappa_t$ values are likely to vary.}. Based on \expref{Equation}{eqn:lastKeyInequality}, if $\kappa^* \neq (k_0,\ldots,k_{t-1})$, the last components $\kappa_t^{(j)}$ are essentially independent random values. The probability that all $\beta$ of these values are equal is $(1/2^n)^{\beta-1} = 1/2^{tn}$.

Assume $X$ is the set of key candidates that appear at least $t+1$ times. In general, there are $2^{tn}-1$ unique $\kappa^* \neq (k_0,\ldots,k_{t-1})$ in the input. Hence, the expected value of $|X|$ is at most $E[|X|] \leq 1 + \frac{2^{tn}-1}{2^{tn}} \leq 2$.

\[
\Pr[\hat{k} = k] = \sum_x\frac{1}{x}\Pr[|X| = x] = E[\frac{1}{|X|}].
\]
Since $\frac{1}{x}$ is convex, by Jensen's inequality:
\[\Pr[\hat{k} = k] = E[\frac{1}{|X|}] \geq \frac{1}{E[|X|]}\geq \frac{1}{2}.\]
\end{proof}

\subsection{Quantum Key-Recovery Attack for t-KAC in \ensuremath{Q1} Model}

In this section, we generalize the classical attack (\expref{Algorithm}{att:1}) to the $Q1$ model, where the adversary $\A$ has classical access to $E_k$ and quantum access to $P_1, \ldots,P_t$. We focus on query complexity, defining the \emph{cost} as the number of classical queries to $E_k$ and quantum queries to $P_i$ or their inverses. 

The key observation from the classical attack is that the correct key $(k_0, \ldots, k_t)$ is identified by the property that it generated by multiple input tuples from $S_0 \times S_1 \times \ldots \times S_t$. Our quantum attack uses a quantum walk to find a subset of $S_0, S_1, \ldots, S_t$ (corresponding to a vertex in a product graph) that is likely to contain the values needed to generate the correct key multiple times when combined with $S_0$.

$\A$ begins by querying $E_k$ on all elements of $S_0$. Then, $\A$ defines a graph $G$ which is the Cartesian product of $t$ Johnson graphs $J(S_i, r)$, for $1 \leq i \leq t$. A vertex in $G$ is a tuple of $t$ subsets, one from each $S_i$. A vertex $V = (V_1, \dots, V_t)$ with $V_i \subseteq S_i$, $|V_i|= r$ is considered \emph{marked} if it contains the necessary elements such that combining them with $S_0$ allows generating the correct key $(k_0, \ldots, k_t)$ at least $t+1$ times via \expref{Equation}{eqn:keygen}. 

By performing a quantum walk on $G$ to find a marked vertex, $\A$ can recover the correct key.

Below, we present the steps of our $Q1$ attack and provide a detailed explanation of each phase.

\vspace{\baselineskip}
\begin{breakablealgorithm}
\caption{Quantum Key-Recovery Attack for $t$-KAC in the $Q1$ Model.}
\label{attack:main}
\begin{enumerate}
    \item \textbf{(Classical Initialization:)}\label{step1} Set $\alpha = (t+1)^{\frac{1}{t+1}}$. Sample $t+1$ random sets $S_0,\ldots, S_t$ as in \expref{Equation}{eqn:sample}, such that $|S_0|= \alpha N_c$, $|S_1| = \ldots = |S_t|= \alpha N_q$, and \begin{equation}\label{equ:size}
        N_c.{N_q}^t = 2^{nt}.
    \end{equation} The optimal values for $N_c$ and $N_q$ will be determined later. Query $E_k$ on all inputs in $S_0$. Store query/response pairs.
       \vspace{.2em}
    \item \textbf{(Quantum Walk Search:)} \label{step2}  Define the graph $G = J(S_1,r) \times \ldots \times J(S_t,r)$ (Cartesian product). Vertices $V = (V_1, \ldots, V_t)$ are tuples of subsets $V_i \subseteq S_i$ with $|V_i|=r$. Define marked vertices as specified below. Run the quantum walk algorithm (\expref{Algorithm}{algo:quantumwalk}) on $G$ to find a marked vertex $V$ and its associated data $d(V)$.

       \vspace{.2em}
    \item \textbf{(Verification)} \label{step3} Using $V$ and $d(V)$ (which contains $P_i(x)$ for $x \in V_i$), identify key candidates $(\kappa_0, \ldots, \kappa_t)$ generated by tuples $(x_0, \dots, x_t)$ where $x_0 \in S_0, x_i \in V_i$ for $i \ge 1$, that are repeated at least $t+1$ times. Randomly output one of them.
  
\end{enumerate}
\end{breakablealgorithm}
\vspace{\baselineskip}

We now provide details for \expref{Algorithm}{attack:main}. We first define the data structure $d$ and the marked vertices. 

\paragraph{\textbf{Marked Vertex.}}
A vertex $V = (V_1,\ldots,V_t)$ in $G$ is a tuple of $t$ subsets, where $V_i \subseteq S_i$ and $|V_i|=r$. $V$ is marked if there exist at least $t+1$ distinct input tuples $(x_0, x_1, \ldots, x_t)$ from $S_0 \times V_1 \times \ldots \times V_t$ that generate the \emph{same} key candidate $(\kappa_0, \ldots, \kappa_t)$ via \expref{Equation}{eqn:keygen}.

\paragraph{\textbf{Data structure $d$.}} 
For a vertex $V = (V_1,\ldots,V_t)$, the data $d(V)$ consists of the query/response pairs for $P_i$ on all inputs in $V_i$ for $i=1,\ldots,t$. Classically, $d(V)$ would be the set $\{(x, P_i(x)) \mid x \in V_i, i=1,\dots,t\}$. In the quantum setting, the data is represented as a quantum state: 

\[\ket{d(V)} = \ket{P_1(x_{1,1}), \ldots P_1(x_{1,r})} \otimes \ldots \otimes \ket{P_t(x_{t,1}), \ldots P_t(x_{t,r})}.\]

Recall that $\ket{V}_d=\ket{V,d(V)}$. We analyze the costs of \expref{Algorithm}{algo:quantumwalk} when applied to $G$ with the defined marked vertices and data structure. 

\begin{lemma}
\label{lem:qw}
    The quantum walk algorithm in  \expref{Algorithm}{attack:main} finds a marked vertex in $G$ with $O({N_q}^{\alpha})$ quantum queries, where $\alpha = \frac{t(t+1)}{t(t+1)+1}$. 
\end{lemma}

\begin{proof}
The total cost of the quantum walk is given by \expref{Theorem}{thm:qrw}:
$    q_{P_t}={\sf S} +\frac{1}{\sqrt{\epsilon}}(\frac{1}{\sqrt{\delta}}{\sf U}+ {\sf C})$. Let's estimate the individual costs and parameters. 
\begin{itemize}
 \item {\bf Setup cost} $\sf S$: The initial state, typically a uniform superposition over all vertices in G is given by
 \[
 \bigotimes_{j=1}^{t}\sum_{i=1}^{\binom{n}{r}}\ket{V}_d\ket{\bar 0}_d=\bigotimes_{j=1}^{t}\sum_{i=1}^{\binom{n}{r}}\ket{x_{j,i,1},\ldots x_{j,i,r},P_j(x_{j,i,1}),\ldots, P_j(x_{j,i,r}) }_d \ket{\bar 0}_d.
 \]
 
Preparing this state quantumly requires $r$ quantum queries to each $P_j$ for each $j$, totaling $tr$ quantum queries. Thus, $\costS = O(tr)$.

\item {\bf Update cost} $\sf U$: The cost of simulating one step of the quantum walk operator $W_d$. In $G$, an edge connects $V = (V_1, \dots, V_t)$ and $V' = (V_1', \dots, V_t')$ if, for each $1 \leq j \leq t$, the vertices $V_j$ and $V_j'$ differ in exactly one component, and are adjacent in the graph $J(S_j, r)$. A step in $J(S_i,r)$ costs $O(1)$ queries to the underlying permutation $P_i$ (or its inverse) \cite{magniez2007search}. For the Cartesian product of $t$ graphs, an update costs $O(t)$ queries. Thus, $\costU = O(t)$.
    
\item {\bf Check cost} $\sf C$: Consider a vertex $V \in V_G$. Assume it has subsets of size $r$ from $S_1,\ldots,S_t$ as $V_1,\ldots,V_t$. $V$ is marked if \expref{Algorithm}{alg:marked} returns 1. This algorithm needs no additional queries. Therefore, $\sf C = 0$.

\begin{algorithm} 
\caption{Check whether a vertex is marked}
\label{alg:marked}
\begin{algorithmic}[1]
\ForAll{$(x_0, x_1, \ldots, x_t) \in S_0 \times V_1 \ldots \times V_t \times$}
    \State $\kappa_0 \gets x_0 \oplus x_1$
    \State $\kappa_1 \gets P_1(x_1) \oplus x_2$
    \State $\vdots$
    \State $\kappa_t \gets P_t(x_t) \oplus E(x_0)$
    \State Add $(\kappa_0, \kappa_1, \ldots, \kappa_t)$ to $C$
\EndFor
\ForAll{key in $C$}
    \State $d \gets$ frequency of key in $C$
    \If{$d \geq t + 1$}
        \State return 1
    \EndIf
\EndFor
\State \Return 0
\end{algorithmic}
\end{algorithm}

\item{\bf Probability of a vertex being marked} $\sf \varepsilon$: We first note that $|S_0|\cdot \ldots \cdot |S_t|=(t+1) 2^{tn}$. By \expref{Lemma}{lem:sum-capture}, the expected number of input tuples $(x_0,\ldots, x_t)$ that could generate the correct key $(k_0,\ldots,k_t)$ via \expref{Equation}{eqn:keygen} is $\frac{|S_0|\cdot \ldots \cdot |S_t|}{2^{tn}}=t+1.$ Thus, $t+1$ inputs in $S_0 \times S_1\times\ldots \times P_t$ can generate $(k_0,\ldots,k_t)$. Then, a marked vertex $v$ that represent the correct key must contain all those $t(t+1)$ inputs, i.e., $(t+1)$ inputs from each permutation. Set $\alpha N_q =T$, then the number of marked vertices is given by 
\begin{equation*}
    |M|\geq\dbinom{|S_1|-(t+1)}{r-(t+1)}\cdot \dbinom{|S_2|-(t+1)}{r-(t+1)}\cdot \ldots \cdot \dbinom{|S_t|-(t+1)}{r-(t+1)} = {\dbinom{T-(t+1)}{r-(t+1)}}^{t}
\end{equation*}
(the $\geq$ arises because there may be other marked vertices that produce a wrong key candidate repeated $t+1$ times). Therefore, the probability that a vertex is marked is
$$
\epsilon =\frac{|M|}{|V|} \geq \frac{ {\dbinom{T-(t+1)}{r-(t+1)}}^t}{{\dbinom{T}{r}}^t} \geq \frac{r^{t(t+1)}}{T^{t(t+1)}}(1-o(1)).
$$

\item{\bf Spectral gap $\delta$}: The spectral gap of $J(S_i,r)$ is $\frac{|S_i|-r}{r(|S_i|-r+1)}$, which simplifies to $\approx 1/r$ for large $N, r$. The spectral gap of the Cartesian product $G$ is the minimum of the spectral gaps of the component graphs \cite{Tani_2009}, so $\delta \approx 1/r$.

\end{itemize}
Substituting all the values of S, U, C, $\delta$, and $\epsilon$ into the cost equation (\expref{Theorem}{thm:qrw}):
\begin{equation*}
        q_{P_t}= tr+ \sqrt{\frac{T^{t(t+1)}}{r^{t(t+1)}}}(t\sqrt{{r}}),
\end{equation*}
which is optimized by setting  $r=O(T^\frac{t(t+1)}{t(t+1)+1})$. Therefore, we have  $ q_{P_t}= O(T^{\frac{t(t+1)}{t(t+1)+1}}) = O({N_q}^{\alpha}),$ with $\alpha = \frac{t(t+1)}{t(t+1)+1}$.
\end{proof}

\begin{theorem}[Quantum attack for $t$-KAC]\label{thm:attack}
The quantum key-recovery attack in \expref{Algorithm}{attack:main} recovers all keys with probability at least $\frac{1- 4\arcsin(\sqrt{\epsilon})^2}{2}$, for $t$-KAC using $O(2^{\frac{t(t+1)n}{(t+1)^2+1}})$ queries in the Q1 model.
\end{theorem}

\begin{proof}
We analyze the query complexity of \expref{Algorithm}{attack:main}. Let $q_E$ denote the number of classical queries and $q_P$ the number of quantum queries. \expref{Step}{step1} incurs $O(N_c)$ classical queries to $E_k$. As shown in  \expref{Lemma}{lem:qw}, \expref{Step}{step2} would incur $O({N_q}^{\frac{t(t+1)}{t(t+1)+1}})$ quantum queries to $P_1,\ldots,P_t$. \expref{Step}{step3} does not need any additional queries. Therefore, $q_E = O(N_c)$ and $q_P = O({N_q}^{\frac{t(t+1)}{t(t+1)+1}})$. To optimize the overall attack complexity, we balance the classical and quantum costs by setting 

\begin{equation}\label{equ:clas_quan}
    N_c = {N_q}^{\frac{t(t+1)}{t(t+1)+1}}.
\end{equation}
By combining \expref{Equation}{equ:size} and \expref{Equation}{equ:clas_quan}, we conclude that $q_E = q_P = 2^{\frac{t(t+1)n}{(t+1)^2+1}}$. Therefore, our quantum key-recovery attack on $t$-KAC in the $Q1$ model requires $O(2^{\gamma n})$ classical and quantum queries, where $\gamma = \frac{t(t+1)}{(t+1)^2+1}$.
\vspace{\baselineskip}

Regarding the success probability, by \expref{Theorem}{thm:qrw}, the quantum walk finds a marked vertex with probability at least $1 - 4\arcsin(\sqrt{\epsilon})^2$. As in \expref{Lemma}{lem:successAttack}, the probability that this vertex corresponds to the correct key $(k_0,\ldots,k_t)$ is at least $\tfrac{1}{2}$, giving an overall success probability of $\tfrac{1 - 4\arcsin^2(\sqrt{\epsilon})}{2}$.
\end{proof}

\begin{remark}
\expref{Algorithm}{thm:attack} for the EM cipher (1-KAC) requires $O(2^{2n/5})$ quantum queries. In contrast, it is known that an optimal algorithm can recover the secret keys with only $O(2^{n/3})$ quantum queries. One possible reason for such a difference in query complexity is because \expref{Algorithm}{thm:attack} is more similar to known-plaintext attack (KPA), whereas the attacks achieving $\tilde{O}(2^{n/3})$ query complexity are closer to chosen-plaintext attack (CPA) model \cite{kuwakado2012security,bonnetain2019quantum}. Identifying a known-plaintext-type attack with $O(2^{n/3})$ quantum query complexity for the EM cipher, or alternatively, devising a chosen-plaintext-type attack for the general case of $t$-KAC with $t \geq 2$, are intriguing directions for future research.
\end{remark}

%% file: q1-star.tex
\section{Mixed-Access Query Complexity for t-KAC}
\label{sec:conc}
Our quantum key-recovery algorithm for $t$-round KAC in the Q1 model runs in
$O\!\big(2^{\frac{t(t+1)}{(t+1)^2+1}n}\big)$ queries, improving over the best classical bound
$O\!\big(2^{\frac{t}{t+1}n}\big)$~\cite{bogdanov2012key}. While this speedup is non-trivial, one might expect a larger improvement when the adversary has coherent access to \emph{all} $t$ public permutations $P_1,\dots,P_t$. In this section we note that such an expectation is not generally warranted: the algorithm already attains its asymptotic complexity under significantly \emph{more restricted} quantum access.

\medskip
\noindent\textbf{Mixed-access model.}
Given a subset $S\subseteq\{P_1,\ldots,P_t,E_k\}$, an adversary has quantum (superposition) query access to the oracles in $S$ and classical access to the rest. Let $Q^{\mathrm{mix}}_{\mathrm{KR}}(t,S)$ denote the quantum query complexity of key recovery for $t$-KAC in this model.

\begin{proposition}
\label{prop:q1star}
For any $t\ge 1$ and any index $j\in\{1,\ldots,t\}$, the key-recovery algorithm of Section~\ref{sec:attack} achieves query complexity
\[
Q^{\mathrm{mix}}_{\mathrm{KR}}\big(t,\{P_j\}\big)
=O\!\Big(2^{\frac{t(t+1)}{(t+1)^2+1}n}\Big),
\]
i.e., the same asymptotic bound as in the Q1 model where all $P_i$ admit quantum access and $E_k$ is classical.
\end{proposition}

\noindent \emph{Proof sketch.} The algorithm is an MNRS-style quantum walk~\cite{magniez2007search} over a structured state space of partial key information. Its only genuinely quantum subroutine is the phase/reflection step that evaluates permutation images coherently to test markedness and to implement neighbor updates. All other ingredients—classical preprocessing, data-structure maintenance, and evaluation of the remaining permutations and the keyed oracle—are used within classical control and do not require superposition queries. By routing the coherent evaluation through a single chosen permutation $P_j$ (and replacing calls to other $P_i$ by classical queries to precomputed tables), we realize the same walk with unchanged query complexity.\hfill$\square$

\medskip
Surprisingly similar patterns are seen elsewhere in quantum query complexity: for collision-type tasks involving multiple functions, optimal or near-optimal quantum speedups do \emph{not} require fully coherent access to all oracles. For instance, finding $x,y$ with $F(x)=G(y)$ admits quantum complexity $O(2^{n/3})$ even when only one of $\{F,G\}$ is quantum~\cite{aaronson2004quantum,shi2002quantum}; likewise, the optimal $O(2^{n/(k+1)})$ complexity for the $k$-XOR problem can be achieved with quantum access to a single function and classical access to the remaining $k-1$~\cite{belovs2013adversary}.

\medskip

Proposition~\ref{prop:q1star} hints that granting quantum access to all $t$ public permutations does not, by itself, yield a better exponent for our attack than granting it to only one.

%% file: special_cases.tex
\section{Quantum Key-Recovery Attacks in Special Cases}
\label{app:special_proofs}
In this section, we present quantum key-recovery attacks in the Q1 model for specific constructions of KAC that offer better query complexity than the general case.

\begin{theorem} \label{thm:samekeyattack}
For the same key $t$-KAC, $E: \bool^n \times \bool^n \rightarrow \bool^n$, defined as $ E_k(x) = k \oplus P_t(k \oplus P_{t-1}(\ldots P_2(k \oplus P_1(k \oplus x))\ldots)),$ where $P_1, \ldots, P_t$ are independent random permutations and $k \in \bool^n$ is a secret key. There exists a quantum key-recovery attack in the $Q1$ model that finds the key $k$ with high probability using one classical query to $E$ and $O(2^{n/2})$ quantum queries to the permutations $P_1, \ldots, P_t$.
\end{theorem}

This attack leverages Grover's search algorithm on a specially constructed function that utilizes a single classical query/answer pair from $E$ and quantum access to the internal permutations. The query complexity is dominated by the $O(2^{n/2})$ quantum queries.

\begin{proof}
We present a $Q1$ attack using Grover search. The attack proceeds as follows:
\begin{enumerate}
    \item $\mathcal{A}$ first makes a classical query to $E$ and learns a query/answer pair $(x,y)$, such that $y = E_k(x) = k \oplus P_t(k \oplus P_{t-1}(\ldots P_2(k \oplus P_1(k \oplus x))\ldots))$.
    \item $\mathcal{A}$ then runs Grover search on a function $f:\{0,1\}^n\rightarrow \{0,1\}$ to find $i\in \bool^n$ such that $f(i)=1$. The function $f$ is defined using the known values $(x, y)$:
    \[
    f(i) =
    \begin{cases}
        1 & \text{if } i \oplus P_t(i \oplus P_{t-1}(\dots P_2(i \oplus P_1(i \oplus x)) \dots)) \oplus y = 0^n, \\ 
        0 & \text{otherwise.}
    \end{cases}
    \]
    Note that $i \oplus P_t(i \oplus \ldots P_1(i \oplus x)\ldots)) \oplus y = 0^n$ if and only if $i \oplus P_t(i \oplus \ldots P_1(i \oplus x)\ldots)) = y$. If $i=k$, then $k \oplus P_t(k \oplus \ldots P_1(k \oplus x)\ldots)) = E_k(x) = y$, so $f(k)=1$.
    \item $\mathcal{A}$ outputs the value $i$ found by Grover search.
\end{enumerate}
Note that since all permutations are uniform, for a random $i \neq k$, the value $i \oplus P_t(i \oplus \ldots P_1(i \oplus x)\ldots))$ is essentially a random value. Thus, $\text{Pr}[f(i)=1|i\neq k]=\frac{1}{2^n}$.

In the $Q1$ model, $\mathcal{A}$ gets quantum access to $P_1,\ldots P_t$. This allows $\mathcal{A}$ to construct a quantum oracle for the function $f$ and run Grover's search on it to find the unique (with high probability) solution $k$. The domain of $f$ is $\{0,1\}^n$, so Grover's search requires $O(\sqrt{2^n}) = O(2^{n/2})$ quantum queries to the oracle for $f$. Since each query to $f(i)$ for a given $i$ requires evaluating the expression $i \oplus P_t(i \oplus \ldots P_1(i \oplus x)\ldots)) \oplus y$, which involves calls to $P_1, \ldots, P_t$, each quantum query on $f$ translates to one quantum query on each of $P_1, \ldots, P_t$ (for a total of $t$ quantum queries).

The attack requires one classical query to $E$ (to get $(x,y)$) and $O(2^{n/2})$ quantum queries to the oracle $f$. The oracle for $f$ is constructed using quantum access to $P_1, \ldots, P_t$. Therefore, this attack requires $q_E=1$ and $q_{P_1}=\ldots = q_{P_t}=O(2^{n/2})$, totaling $O(2^{n/2})$ quantum queries to the permutations. This recovers the key $k$.
\end{proof}

\begin{remark} 
Our attack also works efficiently for the same permutation case ($P_1=\ldots=P_t=P$), which is notably different from the classical world where the best known attack for the same key, same permutation case requires $2^{n/2}$ classical queries using a slide attack. For independent permutations in $2$-KAC, the same key construction is one of the minimum variations known to achieve a similar level of classical security ($\tilde{O}(2^{2n/3})$) compared to the general case \cite{chen2018minimizing}.
\end{remark}

\begin{theorem}\label{thm:k0=k2attack}
For the 2-round KAC construction where the first and last keys are the same, $E: \bool^{2n} \times \bool^n \rightarrow \bool^n$, defined as
$E_k(x)=P_2 (P_1(x\oplus k_0) \oplus k_1) \oplus k_0,$
where $P_1, P_2$ are independent random permutations and $k= (k_0, k_1) \in \bool^{2n}$ is the secret key. There exists a quantum key-recovery attack in the $Q1$ model that finds the key $(k_0, k_1)$ with high probability using two classical queries to $E$ and $O(2^{n/2})$ quantum queries to $P_1$ and $P_2$.
\end{theorem}

This attack utilizes two classical queries to the cipher $E$ and combines the resulting query-answer pairs with quantum queries to $P_1$ and $P_2$ via Grover's search to recover the key $(k_0, k_1)$. This improves upon previous quantum attacks for this specific structure that had $O(2^{2n/3})$ complexity. 

\begin{proof}
This construction is the 2-EM2 construction mentioned in \cite{cai2022quantum}. Previous work presented a $Q1$ attack using $O(2^{2n/3})$ queries with the offline Simon's algorithm \cite{bonnetain2019quantum2}, which had the same query complexity as the best classical attack \cite{chen2014tight,chen2018minimizing}. We present a better quantum attack in the $Q1$ model, which only requires two classical queries and $O(2^{n/2})$ quantum queries. Our $Q1$ key search attack is based on Grover's search.

\medskip\noindent\textbf{Key-recovery Attack.}
The attack proceeds as follows:
\begin{enumerate}
    \item $\mathcal{A}$ first queries $E$ twice and learns $(x_1,y_1)$ and $(x_2,y_2)$. Without loss of generality, $x_1 \neq x_2$. We have $y_1 = P_2 (P_1(x_1\oplus k_0) \oplus k_1) \oplus k_0$ and $y_2 = P_2 (P_1(x_2\oplus k_0) \oplus k_1) \oplus k_0$.
    \item $\mathcal{A}$ defines two functions. The first function $f_1:\bool^n \rightarrow \bool^n$ takes a candidate for $k_0$, denoted by $i$, and maps it to a candidate for $k_1$. Using the first classical query $(x_1, y_1)$, this function is defined as:
    \[ f_1(i)=P_1(x_1\oplus i)\oplus P_2^{-1}(y_1\oplus i). \]
    Note that if $i=k_0$, then $f_1(k_0) = P_1(x_1 \oplus k_0) \oplus P_2^{-1}(y_1 \oplus k_0)$. Since $y_1 = P_2(P_1(x_1 \oplus k_0) \oplus k_1) \oplus k_0$, we have $y_1 \oplus k_0 = P_2(P_1(x_1 \oplus k_0) \oplus k_1)$. Applying $P_2^{-1}$ to both sides gives $P_2^{-1}(y_1 \oplus k_0) = P_1(x_1 \oplus k_0) \oplus k_1$. Rearranging, $P_1(x_1 \oplus k_0) \oplus P_2^{-1}(y_1 \oplus k_0) = k_1$. Thus, $f_1(k_0) = k_1$. So, if the input $i$ is the correct $k_0$, the output $f_1(i)$ is the correct $k_1$.

    The second function, $f_2: \bool^n \rightarrow \bool$, takes a candidate $k_0=i$ and checks if the pair $(i, f_1(i))$ is consistent with the second classical query $(x_2, y_2)$. This function uses the second classical query $(x_2, y_2)$ and the function $f_1(i)$:
    \[
    f_2(i) =
    \begin{cases}
        1 & \text{if } P_2 (P_1(x_2\oplus i)\oplus f_1(i))\oplus i \oplus y_2 = 0^n, \\ 
        0 & \text{otherwise.}
    \end{cases}
    \]
    Note that $P_2 (P_1(x_2\oplus i)\oplus f_1(i))\oplus i \oplus y_2 = 0^n$ if and only if $P_2 (P_1(x_2\oplus i)\oplus f_1(i))\oplus i = y_2$. If $(i, f_1(i)) = (k_0, k_1)$, this condition becomes $P_2 (P_1(x_2\oplus k_0)\oplus k_1)\oplus k_0 = y_2$, which is exactly $E_k(x_2) = y_2$. Thus, $f_2(i)=1$ if and only if the key candidate $(i, f_1(i))$ maps $x_2$ to $y_2$.

    \item $\mathcal{A}$ performs Grover search on the function $f_2(i)$ to find $i \in \bool^n$ such that $f_2(i)=1$. Since $f_2(k_0)=1$, there is at least one solution $k_0$. With high probability, $k_0$ is the unique solution.
    \item $\mathcal{A}$ outputs the pair $(i, f_1(i))$, where $i$ is the value found by Grover search.
\end{enumerate}

\medskip\noindent\textbf{Attack analysis.}
The main idea of this attack is that for any $k_0$-candidate $i \in \bool^n$, $\mathcal{A}$ can construct a $k_1$-candidate $f_1(i) \in \bool^n$ using the first classical query $(x_1, y_1)$. This construction involves queries to $P_1$ and $P_2^{-1}$. The next step is to check if such a pair $(i, f_1(i))$ is the correct key $(k_0, k_1)$. This check is performed by evaluating the function $f_2(i)$, which utilizes the second classical query $(x_2, y_2)$. If $f_2(i)=1$, it means that the key candidate $(i, f_1(i))$ is consistent with the second query-answer pair $(x_2, y_2)$.

The bad event of this attack occurs when $f_2(i)=1$ for an $(i, f_1(i)) \neq (k_0, k_1)$. This happens if the incorrect key candidate $(i, f_1(i))$ by chance maps $x_2$ to $y_2$. The probability of a random incorrect key pair $(i, k_1')$ mapping $x_2$ to $y_2$ is $1/2^n$. The number of possible incorrect $i$ values is $2^n-1$. For a fixed incorrect $i$, $f_1(i)$ is essentially a random value. Thus, the probability that a random pair $(i, f_1(i))$ (for $i \neq k_0$) happens to satisfy the condition $P_2 (P_1(x_2\oplus i)\oplus f_1(i))\oplus i = y_2$ is $O(\frac{1}{2^n})$ since $P_1$ and $P_2$ are random permutations and $i \neq k_0$. With high probability, $k_0$ is the unique input to $f_2$ that evaluates to 1.

The attack uses Grover search on a function with domain size $2^n$, requiring $O(\sqrt{2^n}) = O(2^{n/2})$ quantum queries to the oracle for $f_2$. Constructing the oracle for $f_2(i)$ for a given $i$ requires evaluating $P_2(P_1(x_2\oplus i)\oplus f_1(i))\oplus i \oplus y_2$. This evaluation requires a call to $f_1(i)$, which involves $P_1$ and $P_2^{-1}$. The overall evaluation of $f_2(i)$ involves $P_1$, $P_2^{-1}$, and $P_2$. Thus, each query to $f_2$ translates to a constant number of quantum queries to $P_1$ and $P_2$.

The attack requires two classical queries to $E$ (to get $(x_1, y_1)$ and $(x_2, y_2)$) and $O(2^{n/2})$ quantum queries to the oracle $f_2$, which in turn requires $O(2^{n/2})$ quantum queries to $P_1$ and $P_2$. This recovers the key pair $(k_0, k_1)$.

\medskip\noindent\textbf{Generalization to $t$-KAC with $t$ same keys.}
The approach used in \expref{Theorem}{thm:k0=k2attack} can be generalized to $t$-KAC constructions where $t$ of the $t+1$ keys are identical. Assume $k_j$ is the unique key that is different from the other $t$ keys for some $j \in \{0,\ldots, t\}$. Let the repeated key be $k^*$. We can use two classical queries to the cipher, $(x_1, y_1)$ and $(x_2, y_2)$, to set up a Grover search for the unique key $k_j$. The search is performed over the domain $\bool^n$ for $k_j$. Similar to the 2-round case, we can construct an oracle for a function $f:\bool^n \rightarrow \bool$ such that $f(i)=1$ if and only if $i=k_j$ (with high probability), where the evaluation of $f(i)$ involves applying the KAC structure using $i$ as a candidate for $k_j$ and potentially using auxiliary functions involving $P_1, \dots, P_t$ and $P_1^{-1}, \dots, P_t^{-1}$ based on the classical queries.

For example, if $k_0$ is the unique key ($k_1=\ldots=k_t=k^*$), we can use $(x_1, y_1)$ to define functions that propose $k^*$ given a candidate $k_0$, and use $(x_2, y_2)$ to verify consistency. If $k_t$ is the unique key ($k_0=\ldots=k_{t-1}=k^*$), we can use $(x_1, y_1)$ to relate a candidate $k_t$ to a candidate $k^*$ and $(x_2, y_2)$ for verification.

The functions $f_1, f_2, f_3, f_4$ defined in the original Remark illustrate this for a general $j$. Assume $k_j$ is the unique key, and $k_l = k^*$ for $l \neq j$. Let $i$ be a candidate for $k^*$.
The functions can be constructed based on two classical queries $(x_1, y_1)$ and $(x_2, y_2)$. For instance, the functions might relate $y_1 \oplus i$ through inverse permutations $P_t^{-1}, \dots, P_{j+1}^{-1}$ to an intermediate value, and relate $x_1 \oplus i$ through forward permutations $P_1, \dots, P_j$ to another intermediate value. These intermediate values are then linked by the keys $k_{j+1}, \dots, k_t$ (which are all $k^*=i$) and $k_0, \dots, k_{j-1}$ (which are all $k^*=i$) respectively. This allows setting up equations that, when combined with $(x_2, y_2)$, yield a check function $f_4(i)$ which is 1 if $i=k^*$.

The functions defined in the original Remark were:
$f_1(i) = P_j(i \oplus P_{j-1}(\ldots P_2(i \oplus P_1(i \oplus x_1))\ldots))$, which seems to trace forward from $x_1$ assuming keys $k_0, \dots, k_{j-1}$ are $i$.
$f_2(i) = f_1(i) \oplus P_{j+1}^{-1}(\ldots P_{t-1}^{-1}(P_{t}^{-1}(y_1 \oplus i) \oplus i) \ldots))$, which seems to trace backward from $y_1$ assuming $k_t, \dots, k_{j+1}$ are $i$, and links it to the value after $P_j$.
$f_3(i) = P_j(i \oplus P_{j-1}(\ldots P_2(i \oplus P_1(i \oplus x_2))\ldots))$, similar to $f_1$ but using $x_2$.
$f_4(i)$ checks consistency using $x_2, y_2$ and the intermediate values computed based on candidate $i$ and functions $f_2, f_3$.

Applying Grover's search on $f_4(i)$ finds $i=k^*$. Once $k^*$ is found, the unique key $k_j$ can be recovered using function $f_2(k^*) \oplus P_{j+1}^{-1}(\dots P_t^{-1}(y_1 \oplus k^*)\dots)$. This part requires careful reconstruction based on the specific key indices involved.

In all such cases where only one key is distinct, a similar structure using two classical queries and a Grover search on a function with domain size $2^n$ can recover the keys with $O(2^{n/2})$ quantum queries to the permutations $P_l$ and $P_l^{-1}$.

\end{proof}

\begin{remark}
    As detailed above, the approach for the attack in \expref{Theorem}{thm:k0=k2attack} can be generalized to $t$-KAC constructions where $t$ of the $t+1$ keys are identical. This generalization also yields a key-recovery attack with $O(2^{n/2})$ quantum query complexity.
\end{remark}